\documentclass{article}

\usepackage[preprint]{neurips_2026}

\usepackage[utf8]{inputenc} 
\usepackage[T1]{fontenc}    
\usepackage{hyperref}       
\usepackage{url}            
\usepackage{booktabs}       
\usepackage{amsfonts}       
\usepackage{nicefrac}       
\usepackage{microtype}      
\usepackage{xcolor}         

\usepackage{algorithm}
\usepackage{algorithmic}

\usepackage{amsmath}
\usepackage{amssymb}
\usepackage{mathtools}
\usepackage{amsthm}

\theoremstyle{plain}
\newtheorem{theorem}{Theorem}[section]

\theoremstyle{definition}
\newtheorem{definition}[theorem]{Definition}

\theoremstyle{remark}
\newtheorem{remark}[theorem]{Remark}

\usepackage{url}
\usepackage{amsmath,amssymb,amsbsy}
\usepackage{amsthm}
\numberwithin{equation}{section}
\usepackage{mathdots}
\usepackage{paralist}
\usepackage{xcolor}
\usepackage{hyperref}
\usepackage{color}
\usepackage{tikz}

\usepackage{subcaption}

\usepackage{dsfont}
\usepackage{bm}
\usepackage{graphicx}
\graphicspath{{pictures/}}
\DeclareGraphicsExtensions{.pdf,.png,.jpg,.gif}

\usepackage{stackengine}

\usepackage{makecell}

\usepackage[column=0]{cellspace}
\usepackage{graphicx}

\usepackage[shortlabels]{enumitem}

\usepackage{relsize}
\usepackage{float}
\usepackage{enumitem}

\usepackage{comment}
\usepackage{bm}
\usepackage{fancyhdr}
\usepackage{caption}
\usepackage{subcaption}

\usepackage{enumerate}

\newcommand{\R}{\mathbb{R}}

\newcommand{\diag}{\operatorname{diag}}
\DeclareMathOperator*{\argmin}{\text{arg~min}}

\newcommand{\bA}{\boldsymbol{A}}
\newcommand{\bb}{\boldsymbol{b}}

\newcommand{\bD}{\boldsymbol{D}}
\newcommand{\be}{\boldsymbol{e}}

\newcommand{\f}{\boldsymbol{f}}

\newcommand{\bL}{\boldsymbol{L}}

\newcommand{\bq}{\boldsymbol{q}}

\newcommand{\bw}{\boldsymbol{w}}
\newcommand{\bW}{\boldsymbol{W}}
\newcommand{\bx}{\boldsymbol{x}}
\newcommand{\bX}{\boldsymbol{X}}

\newcommand{\bz}{\boldsymbol{z}}

\newcommand{\bell}{\boldsymbol{\ell}}

\newcommand{\calA}{\mathcal{A}}
\newcommand{\calB}{\mathcal{B}}

\newcommand{\calE}{\mathcal{E}}

\newcommand{\calG}{\mathcal{G}}

\newcommand{\calL}{\mathcal{L}}

\newcommand{\calO}{\mathcal{O}}

\newcommand{\calQ}{\mathcal{Q}}

\newcommand{\calV}{\mathcal{V}}

\newcommand{\sign}{\mathrm{sign}}

\graphicspath{{./figs/}}

\newlength{\imgwidth}
\newboolean{twoColVersion}
\setboolean{twoColVersion}{false}
\newcommand{\twoCol}[2]{\ifthenelse{\boolean{twoColVersion}} {#1} {#2} }

\title{Reliable one-bit quantization of \\ bandlimited graph data via single-shot noise shaping}

\author{%
  Johannes Maly \\
  Department of Mathematics \\ 
  Ludwig-Maximilians-Universität \\ 
  and Munich Center for Machine Learning (MCML) \\
  Munich, Germany \\
  \texttt{maly@math.lmu.de} \\
  \And
  Anna Veselovska \\
  Department of Mathematics and Munich Data Science Institute \\ 
  Technical University of Munich \\ 
  and Munich Center for Machine Learning (MCML) \\
  Munich, Germany\\
  \texttt{hanna.veselovska@tum.de} \\
}

\begin{document}

\maketitle

\begin{abstract}
  Graph data are ubiquitous in natural sciences and machine learning. In this paper, we consider the problem of quantizing graph structured, bandlimited data to few bits per entry while preserving its information under low-pass filtering. We propose an efficient single-shot noise shaping method that achieves state-of-the-art performance and comes with rigorous error bounds. In contrast to existing methods it allows reliable quantization to arbitrary bit-levels including the extreme case of using a single bit per data coefficient.
\end{abstract}

\section{Introduction}

In various machine learning applications, the data that have to be processed exhibit a graph structure such as in social networks, web information analysis, and sensor networks \citep{ortega2018graph}.
To reduce data processing costs, research in graph signal processing (GSP) aims to extend well-developed tools from signal processing to graph data by accounting for underlying connectivity \citep{shuman2013emerging,sandryhaila2014big}. This is possible since graph data often has a favorable spectral structure under application of the Graph Fourier Transform \citep{dong2016learning, shuman2016vertex}. An interesting question in this context is whether one can leverage the spectral structure to minimize the error induced by quantizing the data representation to few bits per scalar.  

The interplay between quantization and GSP has been considered from different angles in the past years. \citet{chamon2017finite} study the effect of finite-precision arithmetics on low-pass graph filters. \citet{nobre2019optimized} propose an optimized quantization scheme for reducing the message transmission between nodes when applying polynomial filtering to graph data in a distributed way. In a similar distributed setting,  
\citet{saad2022quantization} examine the noise induced by dithered quantizers when applying different types of filtering to graph data.

Various works have investigated how graph data can be reconstructed from few quantized samples. \citet{di2018optimal,kim2019toward} empirically examine the algorithmic reconstruction of graph data from noisy samples that have been quantized by a memoryless scalar quantizer. \citet{li2021graph,li2022graph} consider optimal task-based quantization by fixing a bandlimited signal on a graph and optimizing the quantization scheme.
\citet{kim2022quantization} proposes a greedy algorithm to select an optimal sampling set on which the discrepancy between a bandlimited ground-truth and its naively quantized counterpart is minimized.

Our work is closely related to the works of \citet{krahmer2023quantization,krahmer2026quantization} and \citet{reingruber2025efficient} which aim at robust quantization without subsampling. \citet{krahmer2023quantization,krahmer2026quantization} use $\Sigma\Delta$-quantization to quantize bandlimited graph data to few bits per entry. \citet{reingruber2025efficient} use local graph Fourier frames to extract a quantized representation. In contrast to the works of \citet{krahmer2023quantization,krahmer2026quantization} and us, \citet{reingruber2025efficient} use overcomplete frames, which leads to a higher-dimensional quantized representation.

Finally, we remark that graph data may exhibit structural properties beyond low-bandwidth behavior \citep{tanaka2020generalized}. In addition, efficient quantization of graph signals plays an important role in enabling scalable training of \emph{Graph Neural Networks (GNNs)} \citep{courbariaux2015binaryconnect,tailor2021degreequant}.


\subsection{Contribution}

In this paper, we propose a novel single-shot noise shaping method for quantizing graph data. It draws inspiration from a recent result of \citet{maly2023simple} on neural network quantization. A key contribution of the present work is to show that the linear quantization framework of \citet{maly2023simple} can be lifted to graph domains by using incoherence to account for the underlying graph geometry. This lifting leads to the removal of restrictions on the realizable bit-levels of the quantization scheme, which existed in prior works on graph data quantization \citep{krahmer2023quantization,krahmer2026quantization}. \\
Our method translates an $N$-dimensional graph structured signal into an $N$-dimensional quantized representation. In contrast to the method of \citep{krahmer2023quantization} --- which is computationally favorable, but restricted to using $\log\log(N)$  bits per entry --- \emph{our approach allows flexible tuning of the bit-level, including the extreme case of using a single bit per data coefficient}. Furthermore, it comes with notably improved theoretical worst-case guarantees and a strongly simplified proof strategy.

\textbf{Roadmap.} Section~\ref{sec:graphPrelim} introduces notation and basic concepts from graph signal processing.
In Section~\ref{sec:MainResults}, we present the proposed single-shot noise shaping (SSNS) quantization method for bandlimited graph data and state our theoretical guarantees, explaining the relation to the work of \citet{maly2023simple} and comparing our guarantees with existing iterative noise-shaping approaches.
Section~\ref{sec:NumericalExperiments} presents numerical experiments on several graph topologies, validating the theoretical scaling and demonstrating practical performance.

\subsection{Notation}
\label{sec:NotationPreliminaries}

In the following, we will abbreviate index sets as $[N] = \{1,\dots,N\}$. We will use bold lowercase and uppercase letters to distinguish vectors and matrices from scalars. The Euclidean and the max-norm of a vector $\bz \in \R^N$ are denoted by $\| \bz \|_2$ and $\| \bz \|_\infty$, respectively. The $\ell_2$-operator norm of a matrix $\bX \in \R^{n_1\times n_2}$ is denoted by $\| \bX \|$. The $(\ell_2\to\ell_\infty)$-operator norm of $\bX$ is denoted by $\| \bX \|_{2,\infty}$. Note that $\| \bX \|_{2,\infty}$ corresponds to the maximum $\ell_2$-norm of the rows of $\bX$. We denote the kernel of $\bX$ by $\ker(\bX)$. When applied to a matrix $\bX^{n\times n}$, the operator $\diag$ extracts its diagonal $\bx \in \R^n$; when applied to a vector $\bx \in \R^n$, $\diag$ returns the corresponding diagonal matrix $\bX \in \R^{n\times n}$. We abbreviate $a \le Cb$, for an absolute constant, by $a \lesssim b$ and write $a \simeq b$ if $a \lesssim b$ and $b \lesssim a$.

\section{Preliminaries on graph data}\label{sec:graphPrelim}

We consider data that lies on an undirected graph $\calG = (\calV,\calE,\bW)$, where $\calV$  is a set of $N$ vertices, $\calE$ is a set of edges, and $W \in \R^{N\times N}$ is a weighted adjacency matrix. Let $d_m:=\sum_{n=1}^N{W_{mn}}$  be the degree of the $m$th vertex. 
The normalized  Laplacian $\calL$ of $\calG$ is defined as  $\calL= \bD^{-1/2}(\bD-\bW)\bD^{-1/2}$, where $\bD=\diag\{(d_1,d_2,\cdots, d_N)^T\}$ is the degree matrix. In our setting, $\calL$ is symmetric, positive semi-definite and thus has an orthogonal eigenbasis ${\bX=[\bx_1, \bx_2, \cdots, \bx_N]}$, where the eigenvectors $\bx_1,\cdots,\bx_N$ have corresponding non-negative eigenvalues $0\le\lambda_1\le\cdots\le \lambda_N$.

Data on the graph can be represented by a function $f\!:\!\calV\! \to\! \R$ defined on vertices of the graph. Note that any $f$ is equivalent to a vector ${\f \in \R^N}$, where the $n$th component of  $\f$ represents the function value at the $n$-th vertex in $\calV$. The eigenvectors and eigenvalues of the Laplacian $\calL$ provide a spectral interpretation of the graph. Indeed, the eigenvalues $\{\lambda_1, \lambda_2, \cdots, \lambda_N\}$ can be interpreted as graph frequencies, while the eigenvectors demonstrate increasing oscillatory behavior as the magnitude of the graph frequencies increases~\cite{shuman2016vertex}. Based on this interpretation, the \emph{Graph Fourier Transform (GFT)} of $\f$ is defined as 
$\widehat{\f}=\bX^T \f \in \R^N$ with entries $\widehat{\f}(\lambda_i)=\langle \f,\bx_i \rangle$. 

In applications, graph data often exhibit intrinsic structure that can be leveraged for efficient data processing, e.g., dominance of low frequencies in the spectral domain \cite{dong2016learning}. Graph data $\f$ are called \emph{bandlimited} if there exists $r \in [N]$ such that the support of the GFT of $\f$ is contained in the frequencies $\{ \lambda_1,\dots,\lambda_{r} \}$
\cite{pesenson2009variational,shuman2016vertex}. 

Let us denote by $\bX_r=[\bx_1, \cdots, \bx_r] \in \R^{N\times r}$ the matrix $\bX$ restricted to the first $r$ columns. Then every $r$-bandlimited function $\f$ can be written as $\f=\bX_r { 
\bm  \alpha}$, for some ${ 
\bm  \alpha\in \R^r}$. By re-normalizing our data, we can always assume that ${\| \f \|_\infty =  1}$. 
The graph's geometry and properties of its bandlimited function subspaces can be quantified by the \emph{graph incoherence} \cite{shuman2016vertex}.\footnote{\citet{krahmer2023quantization} worked with a squared incoherence, i.e., their incoherence parameter $\mu$ is equal to $\mu(\bX_r)^2$ as defined in \eqref{eq:incoh-def}. Note that our definition aligns with the original incoherence of \citet{shuman2016vertex} and leads to tighter bounds since $\mu(\bX_r) \le \mu(\bX_r)^2 \in [1,N/r]$, cf.\ Section \ref{sec:Comparison} below. \label{footnote:Incoherence}}

\begin{definition} \cite{shuman2016vertex}
For the $r$-dimensional subspace of graph functions spanned by  ${\bX_r=[\bx_1, \cdots, \bx_r] \in \R^{N\times r}}$, let ${P_{\bX_r}=\bX_r\bX_r^T}$ be the orthogonal projection onto $\bX_r$. The \emph{incoherence} $\mu(\bX_r)$ of the graph subspace $\bX_r$ is defined via 
\begin{equation}\label{eq:incoh-def}
    \mu(\bX_r)^2 =\frac{N}{r}\max_{1\le i\le N} \| P_{\bX_r} \be_i \|_2^2.
\end{equation}
\end{definition}

It has been shown for various types of random graphs that $\mu$ is small with high probability if the respective graph is sufficiently large \cite{dekel2011eigenvectors,brooks2013non}. This indicates that 
the graph Laplacian eigenvectors
are well-spread in general.

\section{Quantizing graph data}\label{sec:MainResults}

We are interested in quantizing bandlimited graph data. Given an $r$-bandlimited graph function $\f\in \R^N$, we aim to construct a quantized representation $\bq\in\calA^N$ that uses fewer bits per entry, but preserves the important information encoded in $\f$. Motivated by imaging applications \citep{floyd1976adaptive,fornasier2016consistency, huang2021robust,krahmer2023enhanced, krahmer2025mathematics,lyu2023sigma, ehler2021curve}, the resulting quantization error $\calQ\calE$ can be quantified by the $\ell_2$-distance between $\f$ and $\bq$ under the action of a low-pass graph filter $\bL:=[\bell_1\ldots,\bell_N]\in \R^{N \times N}$:
\begin{align*}
    \calQ\calE_{\bL}(\f,\bq) = \| \bL(\f - \bq) \|_2.
\end{align*}

We will focus here on the brick-wall filter $\bL_r = \bX_r\bX_r^T$, for which $\bL_r\f = \f$, and call $\f_{\bq} = \bL_r\bq$ a quantized representative of $\f$.
For $K\geq 1$, we define \emph{midrise} alphabets having $2K$ elements, as sets of the form
\begin{equation}\label{eq-alphabet-midrise}
    \mathcal{A}=\{\pm (k-1/2)\delta: 1\leq k\leq K, k\in\mathbb{Z}\},
\end{equation}
where $\frac{\delta}{2}>0$ denotes the quantizer resolution on the range $[-K\delta,K\delta]$ of $\calA$. The simplest example is the 1-bit alphabet $\{-1, 1\}$ with (full resolution) range $[-2,2]$ and quantizer resolution $1$. 

The \emph{memoryless scalar quantizer} (MSQ) associated with alphabet $\mathcal{A}$ is given by ${Q}:\mathbb{R}\rightarrow \mathcal{A}$ defined as
\begin{equation}\label{eq-MSQ}
    {Q}(z):=\arg\min_{p\in\mathcal{A}}|z-p|.
\end{equation}
For $\calA =\{-1,1\}$ the MSQ $Q$ is given by the two-valued $\sign$-function
\begin{align*}
    \sign(x) = \begin{cases} 1 & x \ge 0 \\ -1 & x < 0. \end{cases}
\end{align*}
When applied to vectors, the MSQ $Q$ acts entry-wise. We will call $Q$ a $B$-bit quantizer if $|\calA| = 2^B$. For such a quantizer with target range $[-c,c]$, $c > 0$, the quantizer resolution of the midrise alphabet in \eqref{eq-alphabet-midrise} satisfies~$\delta~\simeq~c2^{-B}$.

\subsection{Single-shot noise-shaping (SSNS)}

To quantize band-limited graph data $\f$, we will draw inspiration from recent work on neural network quantization \cite{maly2023simple}. The authors propose to combine a simple neuron preprocessing algorithm with plain application of MSQ to quantize neurons in feedforward networks while preserving the input-output relation. It turns out that a similar strategy can be used for reliable quantization of bandlimited graph data. Compared to the iterative noise-shaping quantizer of \citet{krahmer2023quantization,krahmer2026quantization} this approach can be interpreted as a single-shot noise-shaping quantizer, which is also reflected in the name of our method.

\begin{algorithm}[tb]
   \caption{Vector preprocessing \cite{maly2023simple}}
   \label{alg:Preprocessing}
    \begin{algorithmic}
       \STATE {\bfseries Input:} $\bX \in \R^{r\times N}$ ($r < N$), $\bz_0 \in \R^N$, and $c \ge \|\bz_0\|_\infty$
       \vspace{2mm}
       \STATE Initialize $k=0$ and
       \vspace{-2mm}
       $$J_0 = \{ i \in [N] \colon \text{the $i$-th column of $\bX$ is zero} \}$$
       \STATE Define $\bb \in \R^N$ via $\bb_{J_0^c} = \boldsymbol{0}$ and $b_i = c - (z_0)_i$, for $i\in J_0$, such that $\bb \in \ker(\bX)$ and $|(z_0)_i + b_i| = c$ for $i \in J_0$.
       Replace $\bz_0$ with $\bz_0 + \bb$.

       \REPEAT
       \STATE Compute $\bb \in \ker_{J_k}(\bX)$, $\bb \neq \boldsymbol{0}$
       \STATE Compute $\alpha \in \R$ with $\| \bz_k + \alpha \bb \|_\infty = c$
       \STATE $\bz_{k+1} \leftarrow \bz_k + \alpha \bb \in \R^{N}$
       \STATE $J_{k+1} \leftarrow J_k \cup \{ i \in [N] \colon |(\bz_{k+1})_i| = c \}$
       \STATE $k \leftarrow k+1$
       \UNTIL{$\| |\bz_k| - c\boldsymbol{1} \|_0 \le r $}
       \STATE $k_{\text{final}} = k$
      \vspace{2mm}
		\STATE \textbf{Output:} $\bz_{k_{\text{final}}}$ for which $\bX \bz_{k_{\text{final}}} = \bX\bz_0$, $\| \bz_{k_{\text{final}}} \|_\infty = c$, and $\| |\bz_{k_{\text{final}}}| - c \boldsymbol{1} \|_0 \le r$
    \end{algorithmic}
\end{algorithm}

We recall the preprocessing algorithm of \citet{maly2023simple} tailored to the graph setting in Algorithm \ref{alg:Preprocessing}.  The main goal of this preprocessing step is to obtain an alternative representation of 
$\f$ that preserves its low-frequency information while simultaneously minimizing the quantization error when followed by MSQ. While not explicitly mentioned by \citet{maly2023simple}, Algorithm \ref{alg:Preprocessing} is closely related to the Gram-Schmidt walk of \citet{bansal2018gram}. Similar methods appeared far earlier in articles on ``integer-making'' theorems \citep{beck1981integer} and balanced sampling \citep{deville2004efficient}.

Figure \ref{fig:Preprocessing_Illustration} illustrates in $\R^3$ how Algorithm \ref{alg:Preprocessing} proceeds. Intuitively, in each iteration $k$ it constructs a non-zero kernel element $\bb_k$ of $\bX$ that is zero in all entries in which the current iterate $\bz_k$ already achieves the target $\ell_\infty$-norm $c$, and then adds a scaled version of $\bb_k$ to set at least one entry of $\bz_k$ to $\pm c$ while preserving the overall $\ell_\infty$-norm. Since in each step at least one entry of $\bz_0$ is set to $\pm c$, Algorithm \ref{alg:Preprocessing} terminates after at most $N-r$ steps.
Note that in Algorithm \ref{alg:Preprocessing} we denote the kernel of $\bA \in \R^{n_1\times n_2}$ restricted to $J \subset [n_2]$ by
\begin{align}
\label{eq:KerJdef}
    \ker_J(\bA) := \{ \bb \in \ker(\bA) \colon b_i = 0 \text{ if } i \in J \}
\end{align}
and that $\ker_{J_k}(\bX_k) \neq \{\boldsymbol{0}\}$ as long as $\| |\bz_k| - c\boldsymbol{1} \|_0 > r$. Since $\ker_J(\bA) \subset \ker(\bA)$, all iterates $\bz_k$ preserve the low-frequency information $\bX_r\bz_k = \bX_r\bz_0$ when Algorithm \ref{alg:Preprocessing} is applied to $\bX = \bX_r$.

To compute a quantized representation $\bq$ of $\f$, we first apply Alorithm \ref{alg:Preprocessing} to $\bX = \bX_r^T \in \R^{r\times N}$ with $\bz_0 = \f$, and $c = \| \f \|_\infty = 1$. This yields a new representation $\widehat{\f} \in \R^N$ of $\f$ for which $\bX_r^T \widehat{\f} = \bX_r^T \f$, $\| \widehat{\f} \|_\infty = 1$, and
\begin{align}
\label{eq:RelevantComponents}
    |\{ i \in [N] \colon |\widehat{f}_i| \neq 1 \}| \le r.
\end{align}
After preprocessing, all but at most $r$ entries of $\widehat{\f}$ lie exactly on the boundary of the unit $\ell_\infty$-ball and therefore attain the extreme values $\pm 1$.
This property is crucial: when applying a midrise quantizer with an alphabet including $\pm 1$, at most $r$ unsaturated entries
can contribute nonzero quantization error, while the remaining entries are already elements of the quantization alphabet.

\begin{figure}[h]
    \centering
    \scalebox{0.6}{		
\tikzset{every picture/.style={line width=1.25pt}} 

\begin{tikzpicture}[x=0.75pt,y=0.75pt,yscale=-1,xscale=1]
    
    \draw    (330.33,269.67) -- (330.33,51.67) ;
    \draw [shift={(330.33,49.67)}, rotate = 90] [color={rgb, 255:red, 0; green, 0; blue, 0 }  ][line width=0.75]    (10.93,-3.29) .. controls (6.95,-1.4) and (3.31,-0.3) .. (0,0) .. controls (3.31,0.3) and (6.95,1.4) .. (10.93,3.29)   ;
    \draw    (199,159.67) -- (476.33,159.67) ;
    \draw [shift={(478.33,159.67)}, rotate = 180] [color={rgb, 255:red, 0; green, 0; blue, 0 }  ][line width=0.75]    (10.93,-3.29) .. controls (6.95,-1.4) and (3.31,-0.3) .. (0,0) .. controls (3.31,0.3) and (6.95,1.4) .. (10.93,3.29)   ;
    \draw    (421,69.67) -- (229.09,259.59) ;
    \draw [shift={(227.67,261)}, rotate = 315.3] [color={rgb, 255:red, 0; green, 0; blue, 0 }  ][line width=0.75]    (10.93,-3.29) .. controls (6.95,-1.4) and (3.31,-0.3) .. (0,0) .. controls (3.31,0.3) and (6.95,1.4) .. (10.93,3.29)   ;
    \draw  [draw opacity=0][fill={rgb, 255:red, 65; green, 117; blue, 5 }  ,fill opacity=0.17 ][dash pattern={on 0.84pt off 2.51pt}] (152.54,58) -- (273.2,58) -- (552.99,255.67) -- (432.33,255.67) -- cycle ;
    \draw  [dash pattern={on 1.69pt off 2.76pt}][line width=1.5]  (256.33,123.27) -- (301.93,77.67) -- (415,77.67) -- (415,184.07) -- (369.4,229.67) -- (256.33,229.67) -- cycle ; \draw  [dash pattern={on 1.69pt off 2.76pt}][line width=1.5]  (415,77.67) -- (369.4,123.27) -- (256.33,123.27) ; \draw  [dash pattern={on 1.69pt off 2.76pt}][line width=1.5]  (369.4,123.27) -- (369.4,229.67) ;
    \draw  [fill={rgb, 255:red, 0; green, 0; blue, 0 }  ,fill opacity=1 ] (311,112.17) .. controls (311,110.6) and (312.27,109.33) .. (313.83,109.33) .. controls (315.4,109.33) and (316.67,110.6) .. (316.67,112.17) .. controls (316.67,113.73) and (315.4,115) .. (313.83,115) .. controls (312.27,115) and (311,113.73) .. (311,112.17) -- cycle ;
    \draw  [fill={rgb, 255:red, 0; green, 0; blue, 0 }  ,fill opacity=1 ] (304.33,166.17) .. controls (304.33,164.6) and (305.6,163.33) .. (307.17,163.33) .. controls (308.73,163.33) and (310,164.6) .. (310,166.17) .. controls (310,167.73) and (308.73,169) .. (307.17,169) .. controls (305.6,169) and (304.33,167.73) .. (304.33,166.17) -- cycle ;
    \draw  [fill={rgb, 255:red, 0; green, 0; blue, 0 }  ,fill opacity=1 ] (366.33,210.83) .. controls (366.33,209.27) and (367.6,208) .. (369.17,208) .. controls (370.73,208) and (372,209.27) .. (372,210.83) .. controls (372,212.4) and (370.73,213.67) .. (369.17,213.67) .. controls (367.6,213.67) and (366.33,212.4) .. (366.33,210.83) -- cycle ;
    \draw [color={rgb, 255:red, 0; green, 0; blue, 0 }  ,draw opacity=1 ]   (313.83,112.17) -- (307.43,161.35) ;
    \draw [shift={(307.17,163.33)}, rotate = 277.42] [color={rgb, 255:red, 0; green, 0; blue, 0 }  ,draw opacity=1 ][line width=0.75]    (10.93,-3.29) .. controls (6.95,-1.4) and (3.31,-0.3) .. (0,0) .. controls (3.31,0.3) and (6.95,1.4) .. (10.93,3.29)   ;
    \draw [color={rgb, 255:red, 0; green, 0; blue, 0 }  ,draw opacity=1 ]   (307.17,166.17) -- (367.54,209.66) ;
    \draw [shift={(369.17,210.83)}, rotate = 215.77] [color={rgb, 255:red, 0; green, 0; blue, 0 }  ,draw opacity=1 ][line width=0.75]    (10.93,-3.29) .. controls (6.95,-1.4) and (3.31,-0.3) .. (0,0) .. controls (3.31,0.3) and (6.95,1.4) .. (10.93,3.29)   ;
    \draw [color={rgb, 255:red, 65; green, 117; blue, 5 }  ,draw opacity=1 ][line width=0.75]  [dash pattern={on 4.5pt off 4.5pt}]  (152.54,58) -- (432.33,255.67) ;
    \draw    (126,118.67) .. controls (162.63,148.04) and (191.75,147.02) .. (231.14,117.89) ;
    \draw [shift={(232.33,117)}, rotate = 143.13] [color={rgb, 255:red, 0; green, 0; blue, 0 }  ][line width=0.75]    (10.93,-3.29) .. controls (6.95,-1.4) and (3.31,-0.3) .. (0,0) .. controls (3.31,0.3) and (6.95,1.4) .. (10.93,3.29)   ;
    \draw    (539.67,211) .. controls (537.05,234.52) and (521.63,237.55) .. (493.41,238.92) ;
    \draw [shift={(491.67,239)}, rotate = 357.4] [color={rgb, 255:red, 0; green, 0; blue, 0 }  ][line width=0.75]    (10.93,-3.29) .. controls (6.95,-1.4) and (3.31,-0.3) .. (0,0) .. controls (3.31,0.3) and (6.95,1.4) .. (10.93,3.29)   ;
    
    \draw (259.33,88.07) node [anchor=north west][inner sep=0.75pt]    {\Large $\bz_0 = \f$};
    \draw (280.33,170.4) node [anchor=north west][inner sep=0.75pt]    {\Large $\mathbf{z}_{1}$};
    \draw (395,195.4) node [anchor=north west][inner sep=0.75pt]    {\Large $\mathbf{z}_{2} = \widehat{\f}$};
    \draw (498.67,181.07) node [anchor=north west][inner sep=0.75pt]    {\Large $\mathbf{z}_0 +\mathrm{ker}_{J_0}(\mathbf{X})$};
    \draw (75.67,90.73) node [anchor=north west][inner sep=0.75pt]    {\Large $\mathbf{z}_{1} +\mathrm{ker}_{J_1}(\mathbf{X})$};
    \draw (427.33,88.07) node [anchor=north west][inner sep=0.75pt]    {\Large $cB_{\infty }$};

\end{tikzpicture}

}
    \caption{Illustration of Algorithm \ref{alg:Preprocessing} in $\R^3$. The algorithm walks in a kernel direction (green plane) until it hits the boundary of the scaled $\ell_\infty$ ball $cB_\infty$ such that at least one entry of $\bz_k$ is set to $\pm c$. After reducing the feasible kernel directions (dashed green line), the algorithm repeats the procedure until at most $r$ entries of $\bz_k$ are not equal to $\pm c$.}
    \label{fig:Preprocessing_Illustration}
\end{figure}
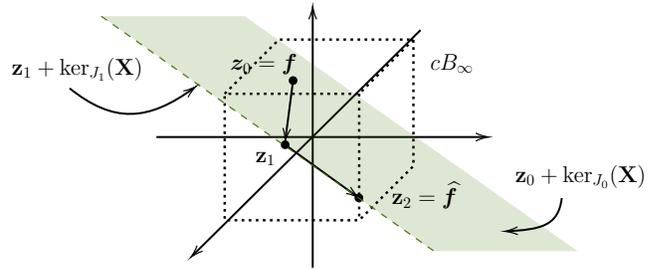

Following the preprocessing step, we apply a $B$-bit quantizer with midrise alphabet $\calA_B = \{ -1, -1 + \frac{2}{2^{B}-1}, \dots, 1 - \frac{2}{2^B-1}, 1\}$, i.e., $K = 2^{B-1}$ and $\delta = 4\cdot 2^{-B}$ to $\widehat{\f}$, which yields the quantized representation $\bq = Q(\widehat{\f})$. The quantization method is summarized in Algorithm \ref{alg:Quantization}.

Due to \eqref{eq:RelevantComponents} and the choice of $\calA$, at most $r$ of the entries of $\widehat{\f}$ will contribute to the quantization error $\calQ\calE_{\bL_r}(\f,\bq)$ in this procedure. We can now bound $\calQ\calE_{\bL_r}(\f,\bq)$.

\begin{theorem}
\label{thm:ErrorBound}
    There exists an absolute constant $C > 0$ such that the following holds. Let $\calG = (\calV,\calE,\bW)$ be an undirected graph. Then, for any $r$-bandlimited $\f \in \R^N$ with $\| \f \|_\infty = 1$, we have that $\bq$ computed via Algorithm \ref{alg:Quantization} satisfies
    \begin{align*}
        \frac{\calQ\calE_{\bL_r}(\f,\bq)}{\| \f \|_2}
        \le C \cdot 2^{-B} \cdot \mu(\bX_r) \cdot \frac{r}{\sqrt{N}},
    \end{align*}
    where $\bL_r = \bX_r\bX_r^T$ is the brick-wall filter of bandwidth $r$, i.e., $\bX_r \in \R^{N\times r}$ contains the eigenvectors of the normalized Laplacian $\calL$ corresponding to the $r$-smallest eigenvalues of $\calL$.\footnote{In the case that the smallest $r$ eigenvalues are not unique, we fix any non-decreasing order and choose the first $r$.}
\end{theorem}

\begin{algorithm}[h!]
   \caption{$B$-bit single-shot noise shaping (SSNS)}
   \label{alg:Quantization}
    \begin{algorithmic}
       \STATE {\bfseries Input:} Graph Laplacian $\calL \in \R^{N\times N}$, graph data $\f \in \R^N$ normalized to $\| \f \|_\infty = 1$, bit-level $B$, target bandwidth $r \in [N]$
       \vspace{2mm}
       \STATE 
       \begin{itemize}[noitemsep, topsep=2pt]
           \item \textbf{Preparation:} Compute $\bX_r \in \R^{N\times r}$ by eigendecomposition of $\calL$
           \vspace{1mm}
           \item \textbf{Preprocessing:} Apply Algorithm \ref{alg:Preprocessing} with $\bX = \bX_r^T \in \R^{r\times N}$, ${\bz_0 = \f}$, and $c {=} \| \f \|_\infty{=} 1$ to obtain the alternative representation $\widehat{\f}$ of $\f$.
           \vspace{1mm}
           \item \textbf{Quantization:} Apply MSQ $Q$ with $B$-bit alphabet $\calA_B$ ($K = 2^{B-1}$ and $\delta = 4\cdot 2^{-B}$) to obtain $\bq{=}Q(\widehat{\f})$
       \end{itemize}
       \vspace{2mm}

		\STATE \textbf{Output:} $B$-bit approximation $\bq$ of $\f$ 
    \end{algorithmic}
\end{algorithm}

\begin{remark}
    Theorem \ref{thm:ErrorBound} is robust against spectral leakage in the following sense: If $\f$ is only of effective low bandwidth, i.e., $\f = \f_r + \f_{\text{high}}$ where $\f_r$ and $\f_{\text{high}}$ are the projections of $\f$ onto the $r$ first and all remaining frequencies, respectively, and we assume that $\|\f_{\text{high}}\|_2 = \varepsilon \| \f \|_2$, for $\varepsilon \in (0,1)$, then Algorithm \ref{alg:Preprocessing} could be applied to $\f_r$ and the bound in Theorem \ref{thm:ErrorBound} would generalize to
    $$ \frac{\mathcal{QE}_L(\f,\bq)}{\| \f\|_2} \le C \cdot 2^{-B} \cdot \mu(\bX_r) \frac{r}{\sqrt{N}} + \varepsilon. $$
    Note that our empirical results suggest robustness of the method to spectral leakage even when the algorithm is directly applied to $\f$ instead of $\f_r$.
\end{remark}

In the proof of Theorem \ref{thm:ErrorBound}, we will use the following insight from \citet{maly2023simple}.

\begin{theorem}[{\citet[Theorem 2.2]{maly2023simple}}]
\label{thm:MalySaab}
    Let $N_0 > m$ be natural numbers, let $\bw \in \R^{N_0}$, and let $\bX \in \R^{N_0 \times m}$. Define the data complexity parameter
    \begin{align*}
        \Gamma(\bX) = \sup_{\substack{I \subset [N_0] \\ |I| = m}} \| \bX^T|_I \|
    \end{align*}
    where $\bX^T|_I \in \R^{m\times m}$ denotes the submatrix of $\bX^T$ with columns indexed by $I$. For $Q$ as defined in Algorithm \ref{alg:Quantization} and $\bq = Q(\widehat{\bw})$ with $\widehat{\bw}$ being the output of Algorithm \ref{alg:Preprocessing} with inputs $\bX^T$, $\bw$, and $\| \bw \|_\infty$, we have that 
    \begin{align*}
        \| \bX^T\bw - \bX^T\bq \|_2 \le C \cdot 2^{-B} \cdot \Gamma(\bX) \cdot \sqrt{m} \cdot \| \bw \|_\infty.
    \end{align*}
\end{theorem}

\begin{proof}[Proof of Theorem \ref{thm:ErrorBound}]
    First note that
    \begin{align*}
        \calQ\calE_{\bL_r}(\f,\bq)
        = \| \bX_r\bX_r^T(\f - \bq) \|_2
        = \| \bX_r^T(\f - \bq) \|_2,
    \end{align*}
    where we used that the columns of $\bX_r$ are orthogonal. Applying Theorem \ref{thm:MalySaab} to the right-hand side with $N_0 = N$, $m=r$, $\bw = \f$ with $\| \f \|_\infty = 1$, and $\bX = \bX_r$ then yields
    \begin{align*}
        \calQ\calE_{\bL_r}(\f,\bq)
        \le C \cdot 2^{-B} \cdot \Gamma(\bX_r) \cdot \sqrt{r}.
    \end{align*}
    Since $\bX_r$ has orthogonal columns, $\Gamma(\bX_r) \le \| \bX_r \| = 1$. Furthermore, since $\f$ is $r$-bandlimited, we obtain
    \begin{align*}
        1
        &= \| \f \|_\infty
        = \| \bX_r \boldsymbol{\alpha} \|_\infty
        \le \| \bX_r \|_{2,\infty} \| \boldsymbol{\alpha} \|_2 \\ 
        &= \mu(\bX_r) \sqrt{\frac{r}{N}} \cdot \| \boldsymbol{\alpha} \|_2, 
    \end{align*}
    for some $\boldsymbol{\alpha} \in \R^r$ with $\| \boldsymbol{\alpha} \|_2 = \| \f \|_2$, where we used that $\mu(\bX_r)^2 = \frac{N}{r} \| \bX_r \|_{2,\infty}^2$ by definition. Since this implies that $\| \f \|_2 \ge \mu(\bX_r)^{-1} \sqrt{\frac{N}{r}}$, we obtain our claimed bound
    \begin{align*}
        \frac{\calQ\calE_{\bL_r}(\f,\bq)}{\| \f \|_2}
        \le \frac{C \cdot 2^{-B} \cdot \sqrt{r}}{\| \f \|_2}
        \le C 2^{-B} \mu(\bX_r) \frac{r}{\sqrt{N}}.
    \end{align*}
\end{proof}

\subsection{Comparison with iterative noise-shaping methods}
\label{sec:Comparison}

Let us compare Theorem \ref{thm:ErrorBound} with existing results of \citet{krahmer2023quantization,krahmer2026quantization} on quantizing bandlimited graph data via iterative noise shaping. Their algorithm iteratively samples with replacement $M$ graph vertices at random, and quantizes the corresponding entry of $\f$ to the alphabet $\calA$ depending on the accumulated quantization error so far. If single vertices have been selected and quantized multiple times, e.g., for $M > N$, these results are accumulated to derive a quantized representation $\bq \in \widetilde \calA^N$ of $\f$. Here $\widetilde \calA$ is an augmented alphabet with $|\widetilde \calA | \lesssim |\calA| \log(N)$ if $M = \calO(N \log(N))$. 

In the extreme case of using a one-bit quantizer with $\calA = \{-1,1\}$ and setting $M = N\log(N)$, Theorem IV.1 in \citet{krahmer2023quantization} shows that
\begin{align}
\label{eq:NoiseQuantizationBound}
    \frac{\calQ\calE_{\bL_r}(\f,\bq)}{\| \f \|_2}
    \lesssim \mu(\bX_r)^2 \frac{r \log(r)}{\sqrt{N \log(N)}}.
\end{align}
It is important to keep in mind that this result requires a total bit budget of $N \log(\log(N))$ bits and that \citet{krahmer2023quantization,krahmer2026quantization} work with the squared incoherence $\mu(\bX_r)^2$, cf.\ Footnote \ref{footnote:Incoherence}. 

If we use the same total bit budget with our method, this corresponds to $B = \log(\log(N))$ since we are quantizing exactly $N$ entries of $\f$. In this case, the bound in Theorem~\ref{thm:ErrorBound} becomes
\begin{align}\label{eq:theor-bound-SSNS}
    \frac{\calQ\calE_{\bL_r}(\f,\bq)}{\| \f \|_2}
    \lesssim \mu(\bX_r) \frac{r}{\sqrt{N}\log(N)}. 
\end{align}

Apparently, the theoretical error bound of our method surpasses the one of \citet{krahmer2023quantization,krahmer2026quantization} in all relevant dependencies ($\mu(\bX_r)$, $r$, and $N$). This is particularly remarkable as our approach is based on single-shot noise shaping before quantization while the work of \citet{krahmer2023quantization,krahmer2026quantization} adaptively shapes the noise during quantization. What is more, the comparison shows that the method of \citet{krahmer2023quantization,krahmer2026quantization} does not allow a final quantization level of only one bit per entry as soon as $\calA$ is augmented in the final representation.

Our experiments in Section \ref{sec:NumericalExperiments} suggest that the worst-case bound in Theorem \ref{thm:ErrorBound} might still be too conservative in general since the empirical performance of the investigated methods considerably surpasses the predicted bounds on several graphs.

\begin{figure*}[tb] 
\centering

\subcaptionbox{Bunny graph\label{fig:bunny-a}}{%
  \includegraphics[width=1.3in]{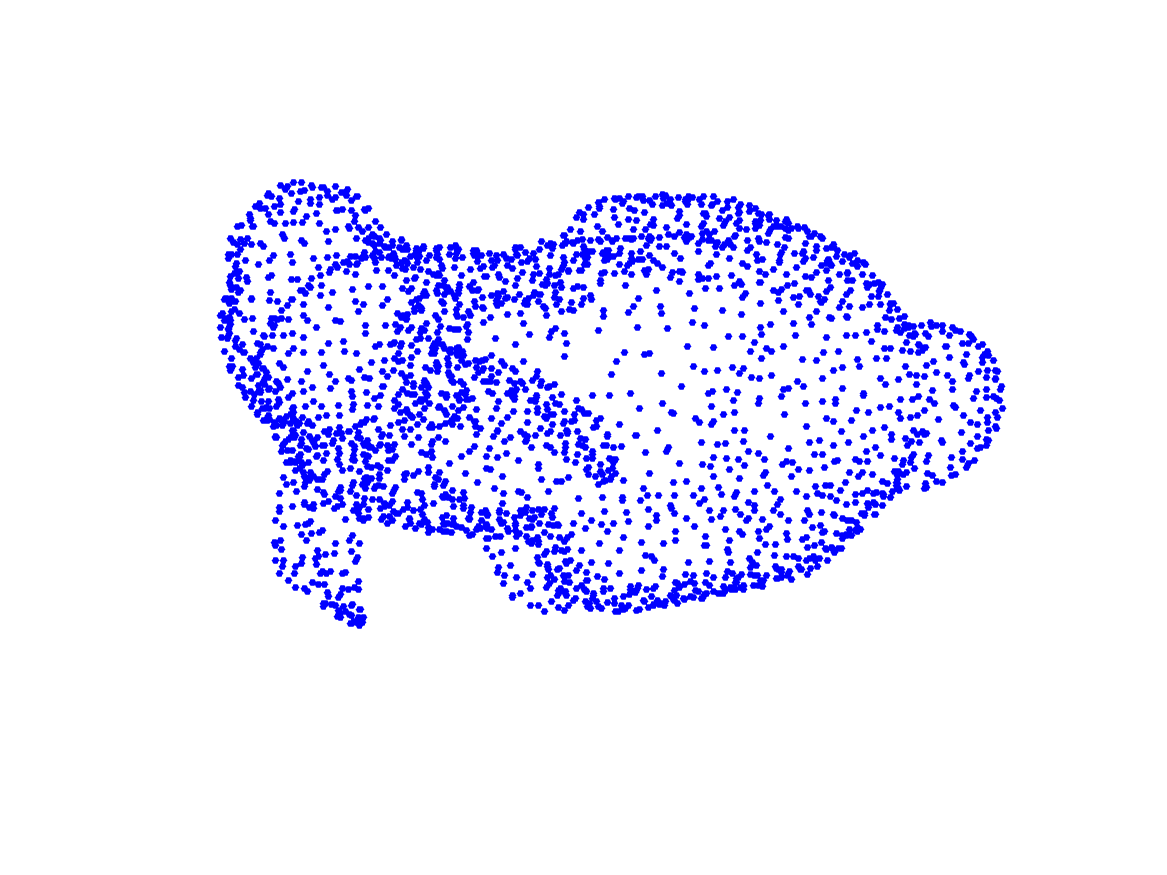}}
\hfill
\subcaptionbox{Relative error\label{fig:bunny-b}}{%
  \includegraphics[width=1.3in]{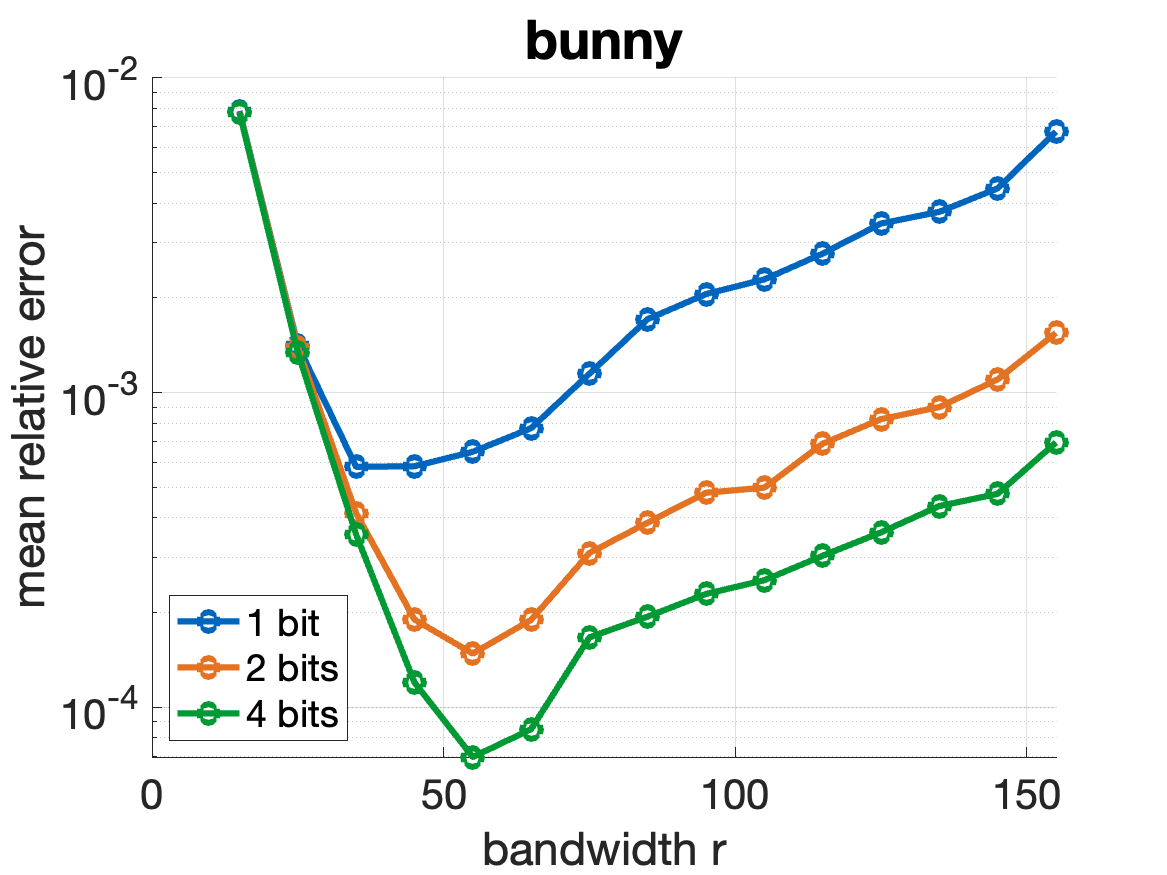}}
\hfill
\subcaptionbox{Swiss roll graph\label{fig:swiss-a}}{%
  \includegraphics[width=1.3in]{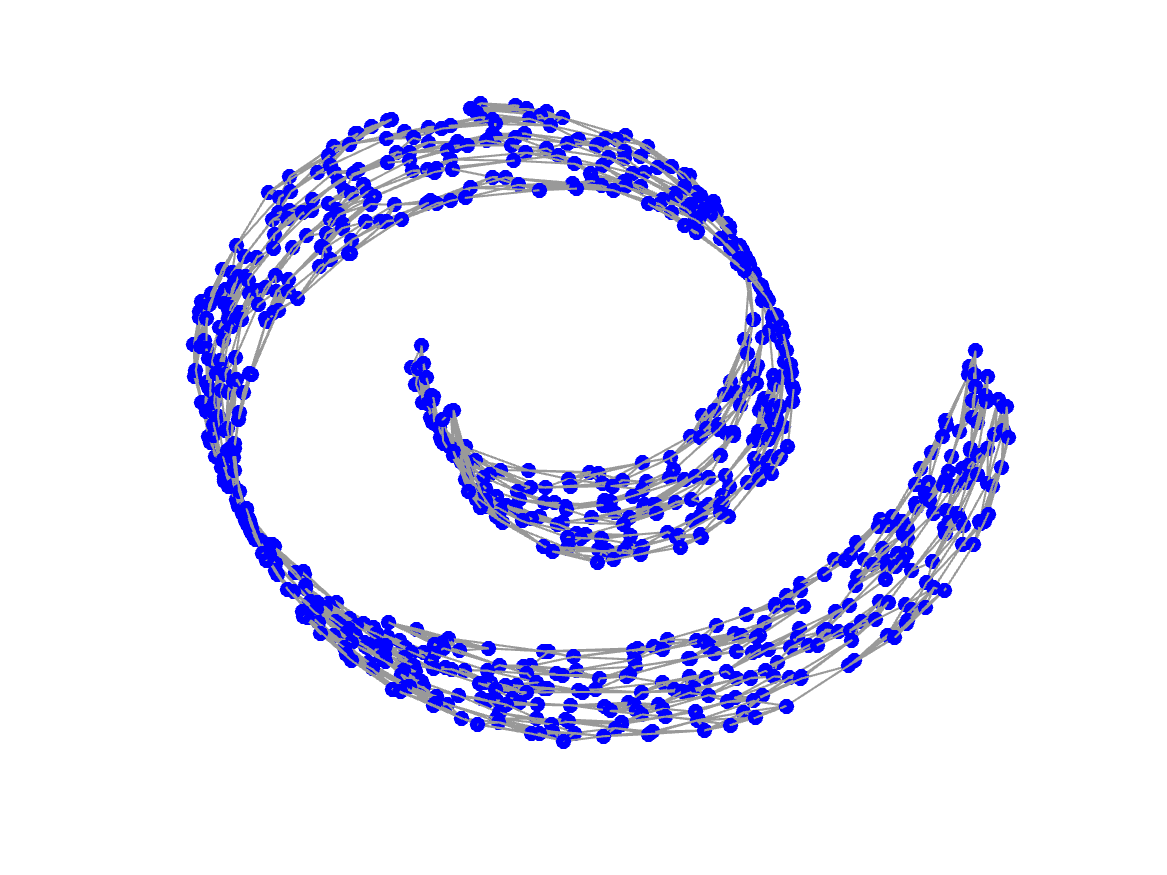}}
\hfill
\subcaptionbox{Relative error\label{fig:swiss-b}}{%
  \includegraphics[width=1.3in]{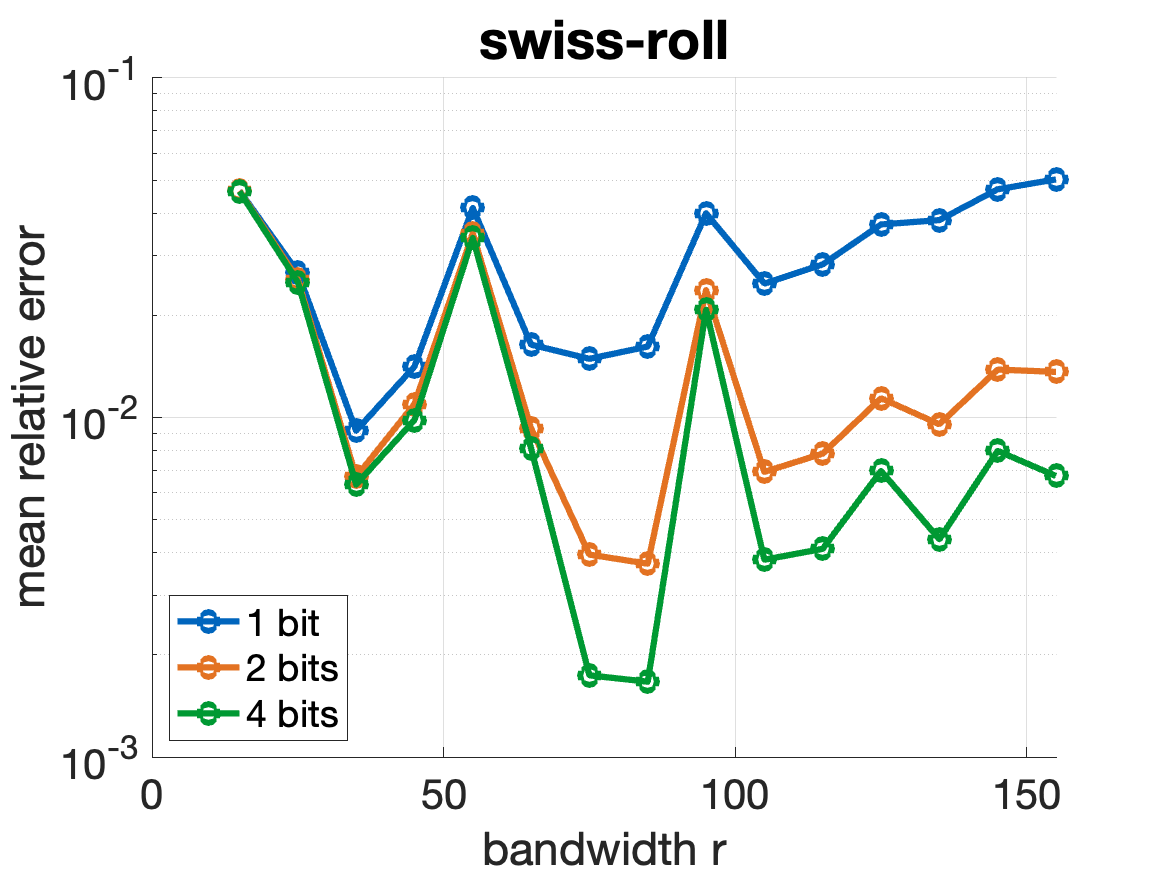}}

\vspace{2mm}

\subcaptionbox{2D grid graph ($30\times30$)\label{fig:grid-a}}{%
  \includegraphics[width=1.3in]{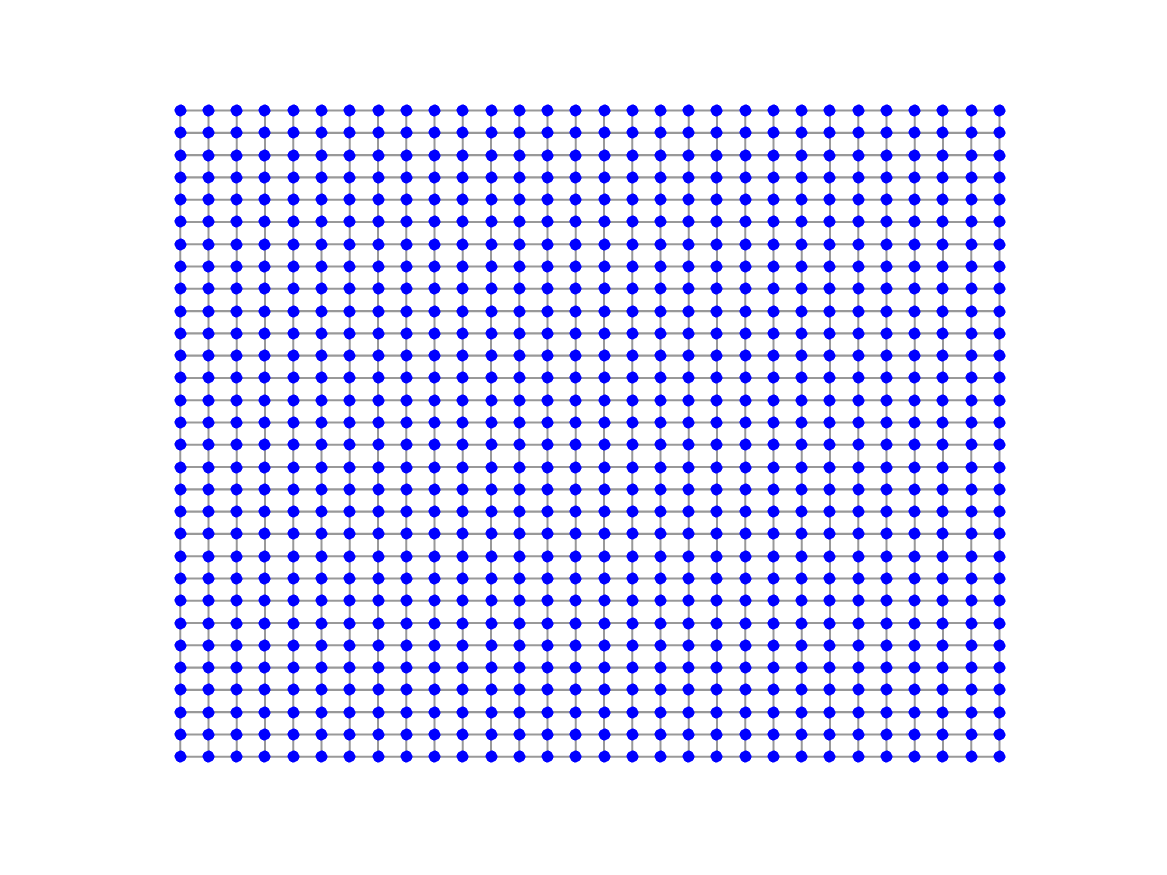}}
\hfill
\subcaptionbox{Relative error\label{fig:grid-b}}{%
  \includegraphics[width=1.3in]{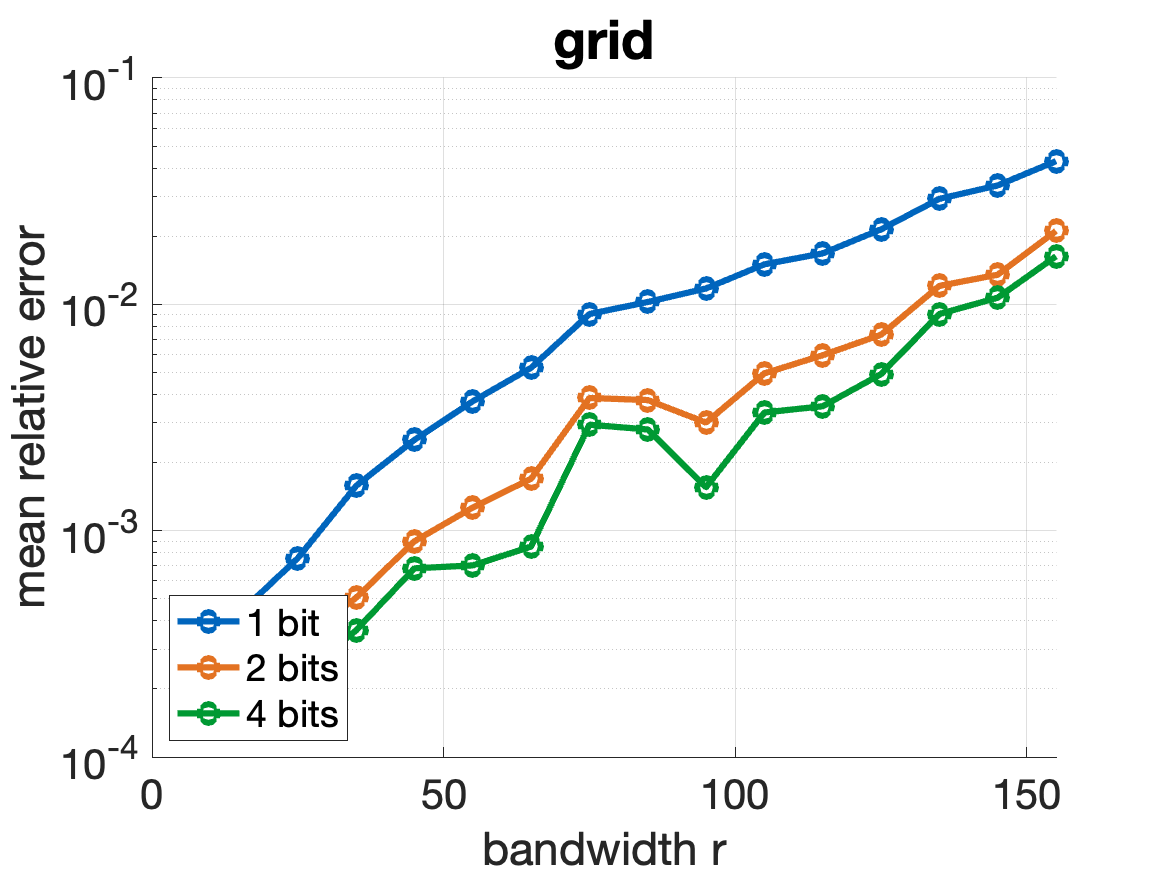}}
\hfill
\subcaptionbox{Ring graph\label{fig:ring-a}}{%
  \includegraphics[width=1.3in]{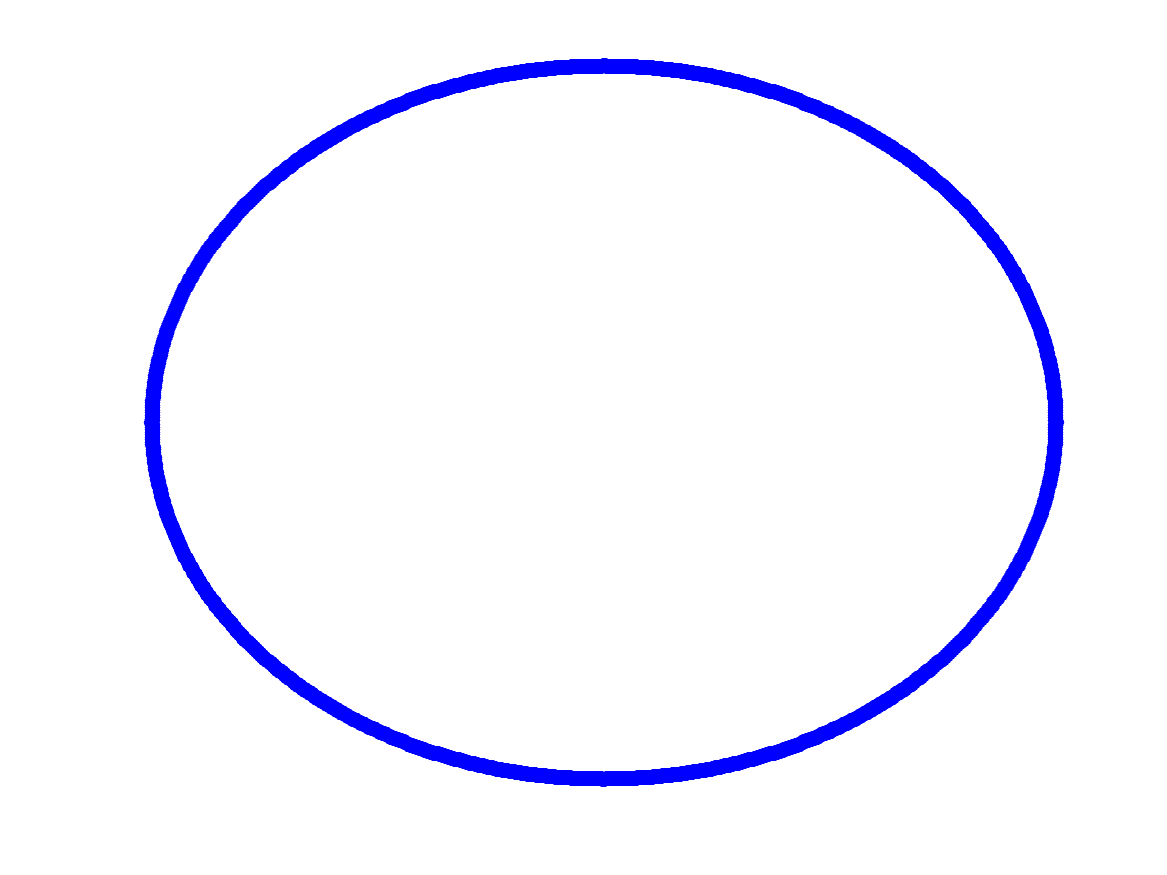}}
\hfill
\subcaptionbox{Relative error\label{fig:ring-b}}{%
  \includegraphics[width=1.3in]{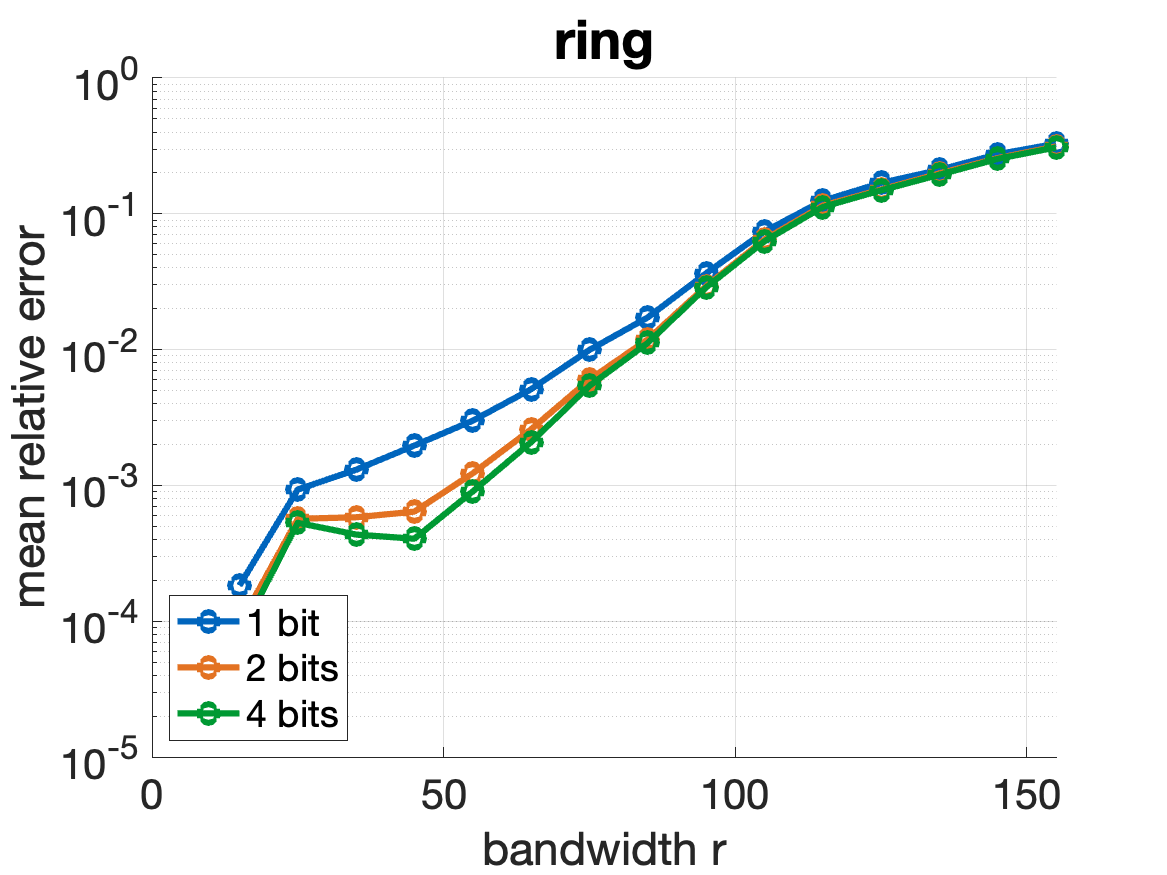}}

\vspace{2mm}

\subcaptionbox{Sensor graph\label{fig:sensor-a}}{%
  \includegraphics[width=1.3in]{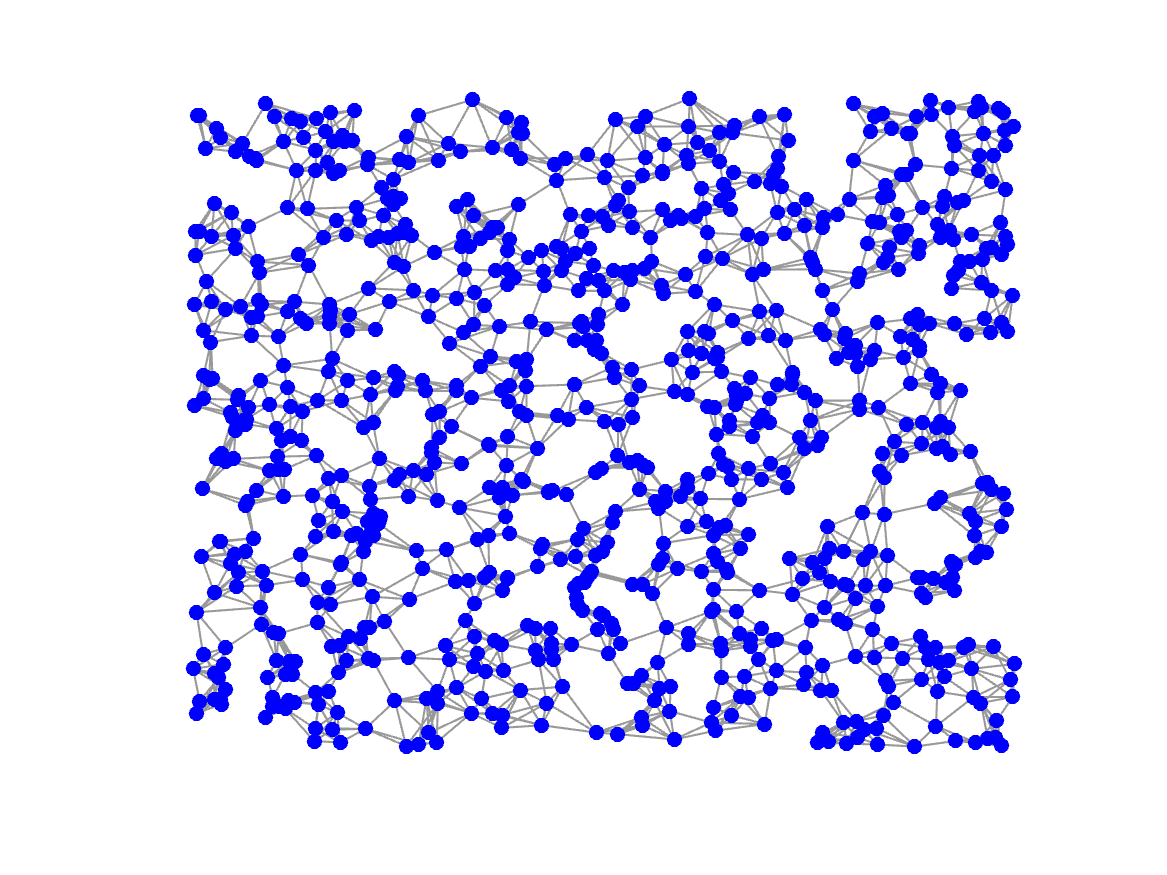}}
\hfill
\subcaptionbox{Relative error\label{fig:sensor-b}}{%
  \includegraphics[width=1.3in]{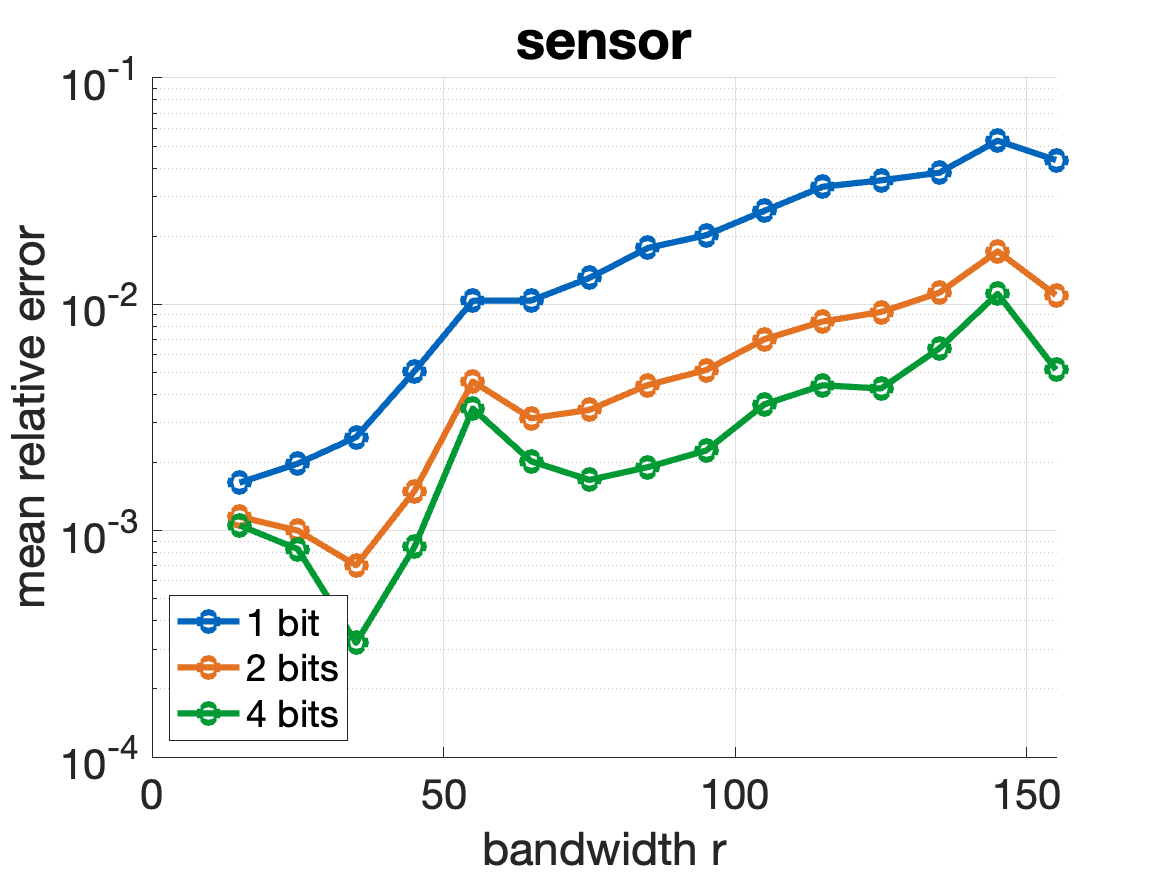}}
\hfill
\subcaptionbox{Minnesota graph\label{fig:minnesota-a}}{%
  \includegraphics[width=1.3in]{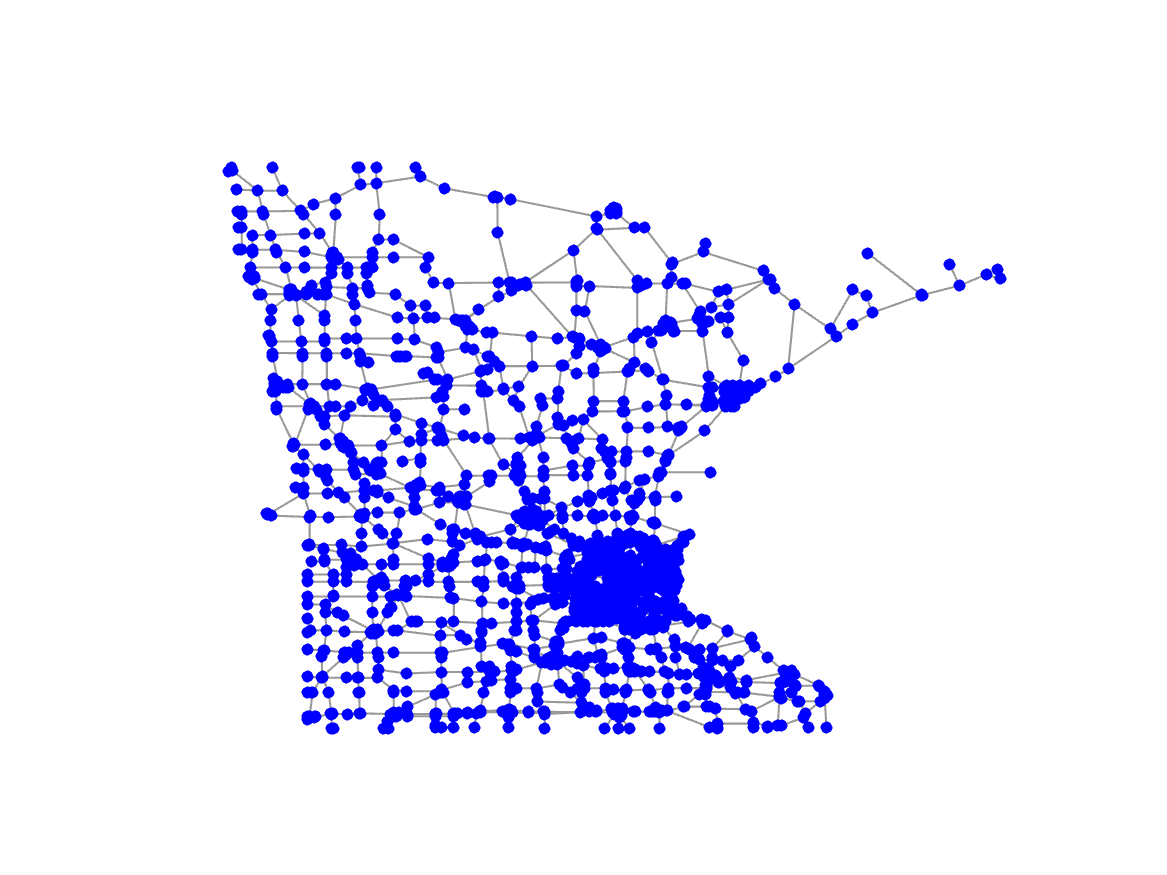}}
\hfill
\subcaptionbox{Relative error\label{fig:minnesota-b}}{%
  \includegraphics[width=1.3in]{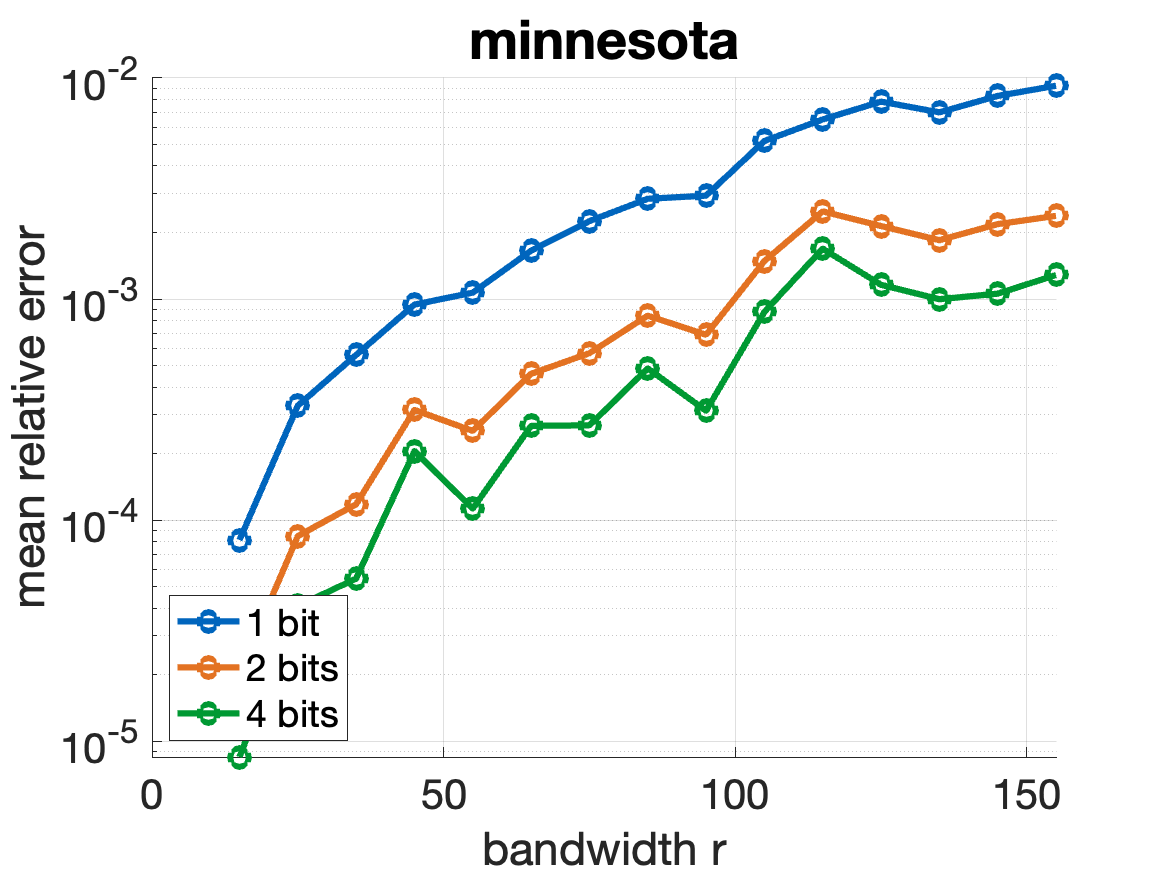}}

\caption{Performance of the proposed SSNS quantization algorithm on different graph structures. 
For each graph (left), the corresponding semilog plot of the relative reconstruction error (right) is shown for different quantization bit budgets $B = 1,2,4$. }
\label{fig:preML_semilog}
\end{figure*}

\section{Numerical Experiments}
\label{sec:NumericalExperiments}

\begin{figure*}[t] 
\centering
\begin{minipage}{0.23\textwidth}
    \centering
    \includegraphics[width=1.15\linewidth]{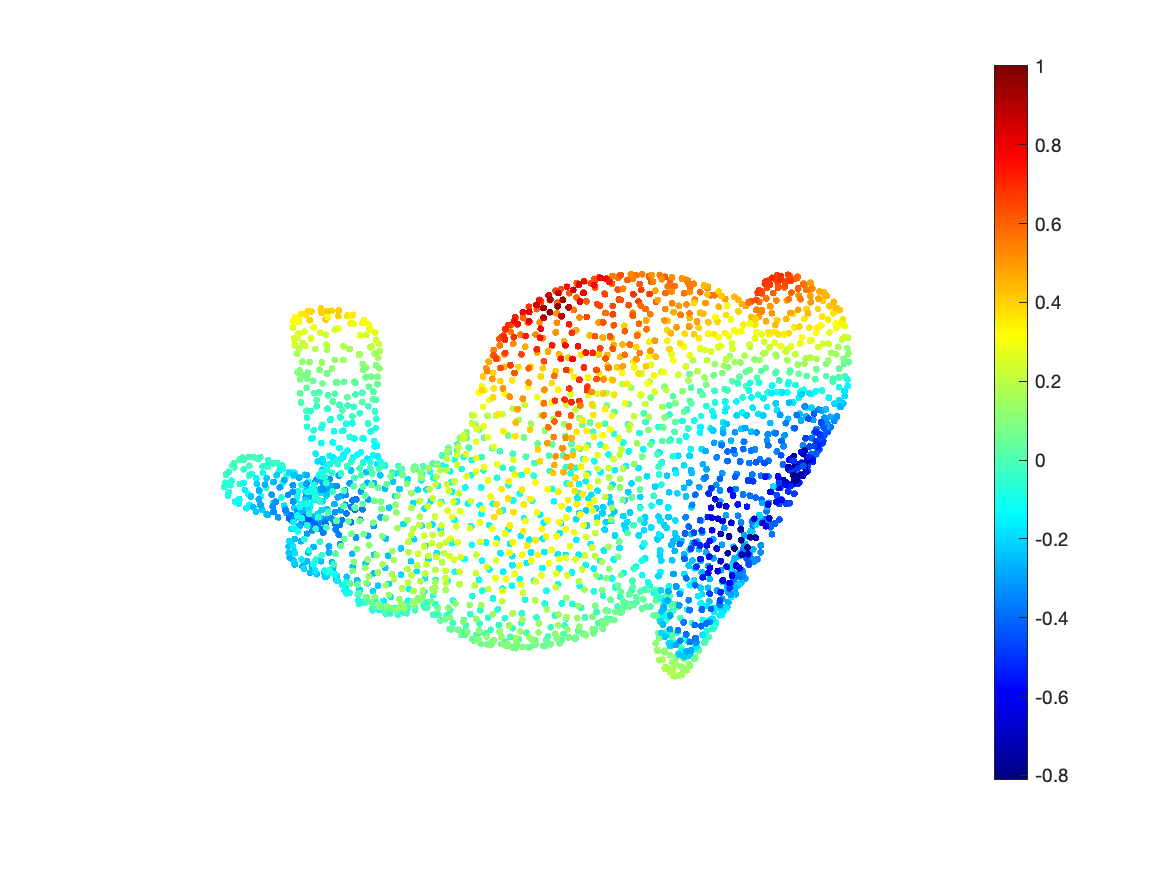}
    \\(a) True data $\f$
\end{minipage}\hspace{5pt}%
\begin{minipage}{0.23\textwidth}
    \centering
    \includegraphics[width=1.15\linewidth]{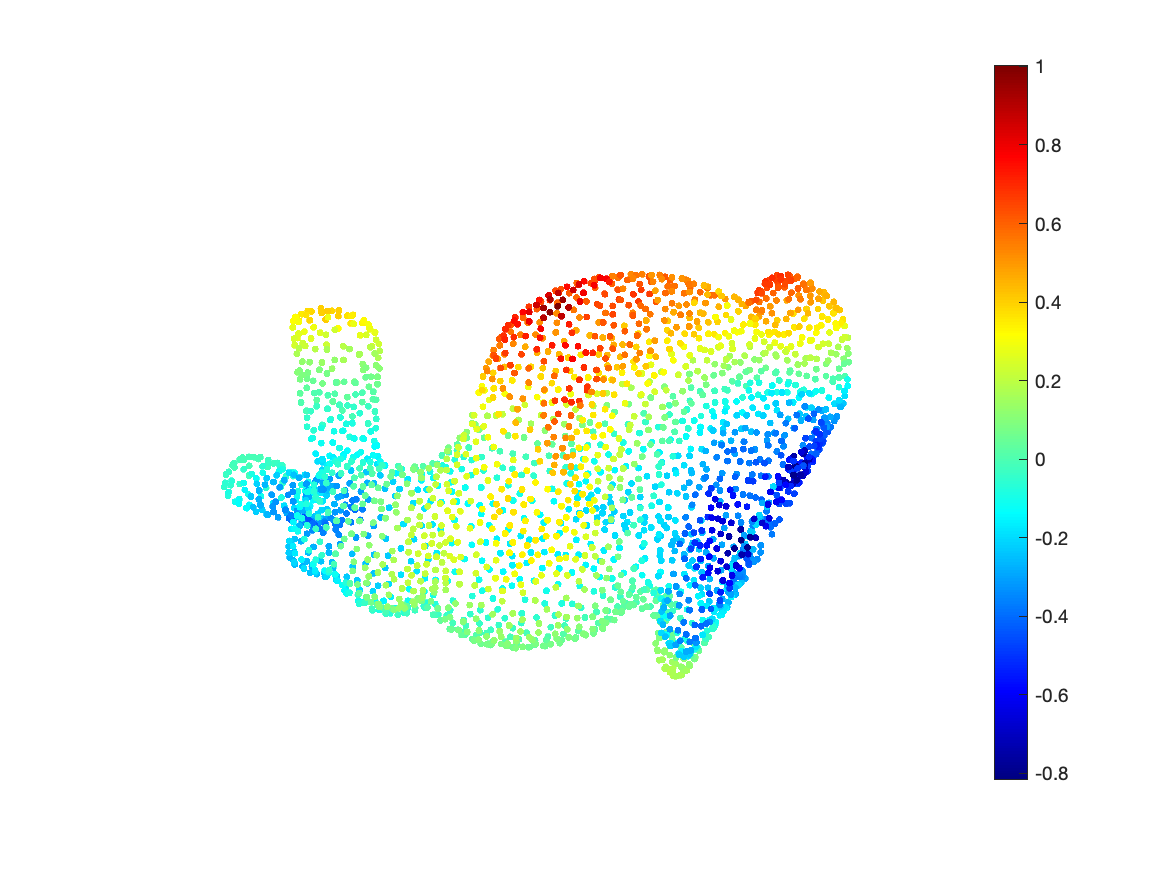}
    \\(b)  $\f_q$ for $B{=}1$ 
\end{minipage}\hspace{5pt}%
\begin{minipage}{0.23\textwidth}
    \centering
    \includegraphics[width=1.15\linewidth]{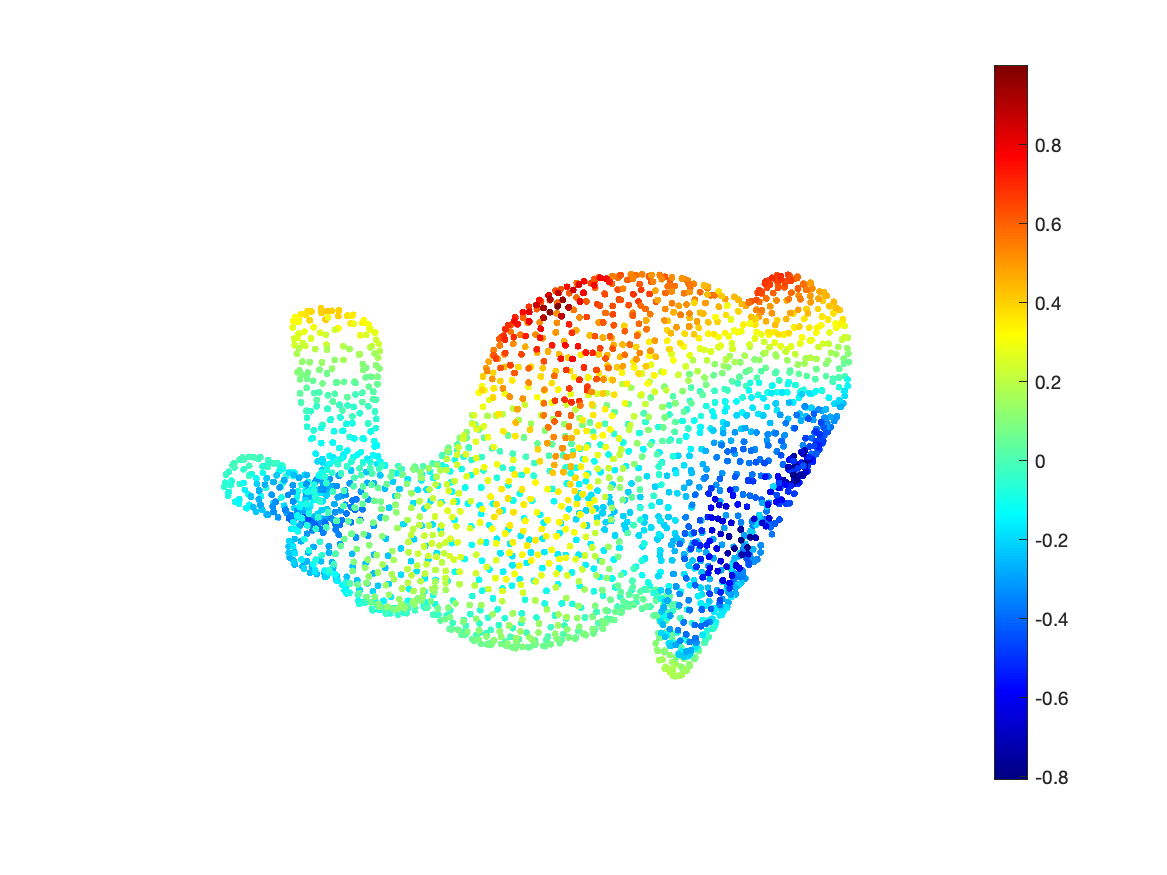}
    \\(c) $\f_q$ for $B{=}2$
\end{minipage}\hspace{5pt}%
\begin{minipage}{0.23\textwidth}
    \centering
    \includegraphics[width=1.15\linewidth]{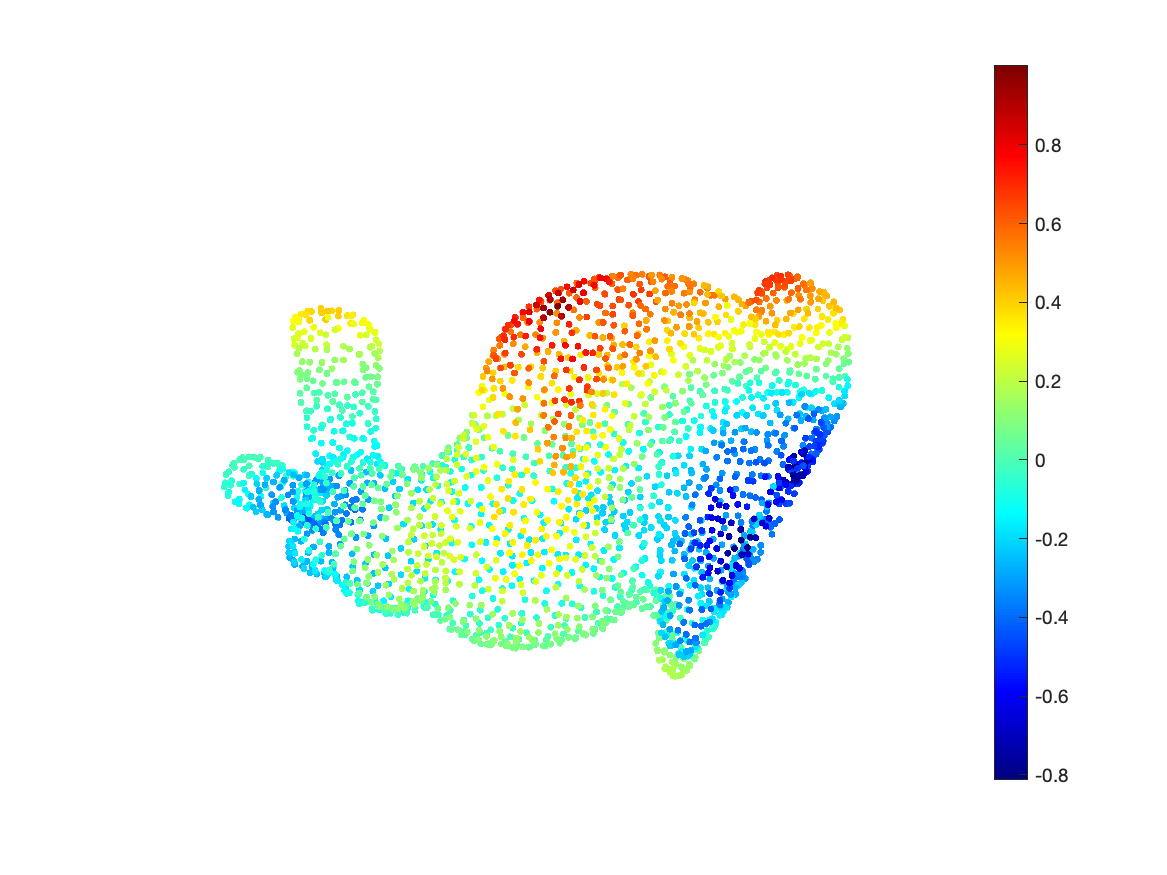}
    \\(d) $\f_q$ for $B{=}4$
\end{minipage}

\vspace{2mm} 

\begin{minipage}{0.23\textwidth}
    \centering
    \includegraphics[width=1.15\linewidth]{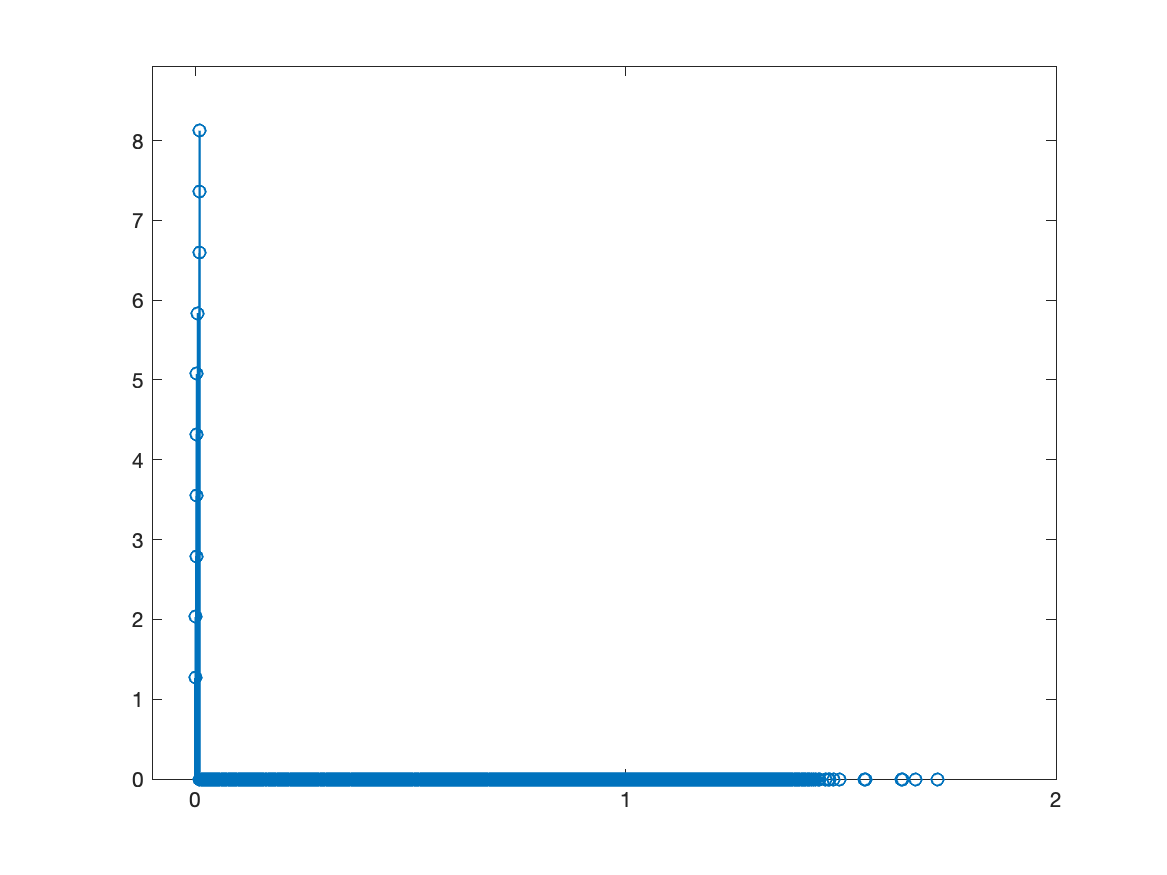}
    \\(e) Data Frequency
\end{minipage}\hspace{5pt}%
\begin{minipage}{0.23\textwidth}
    \centering
    \includegraphics[width=1.15\linewidth]{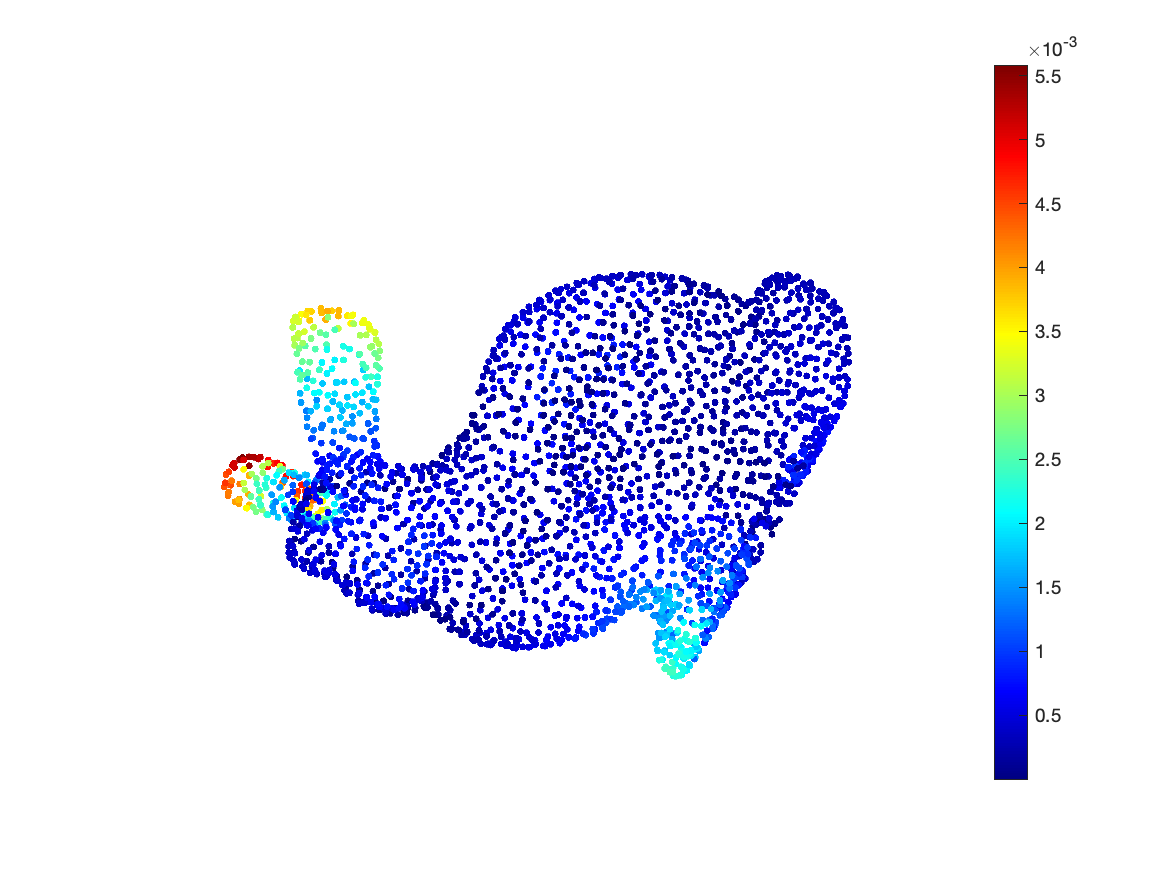}
    \\(f) $\calQ\calE_{\bL}(\f,\bq)=6.4\cdot10^{-5}$ for $B{=}1$
\end{minipage}\hspace{5pt}%
\begin{minipage}{0.23\textwidth}
    \centering
    \includegraphics[width=1.15\linewidth]{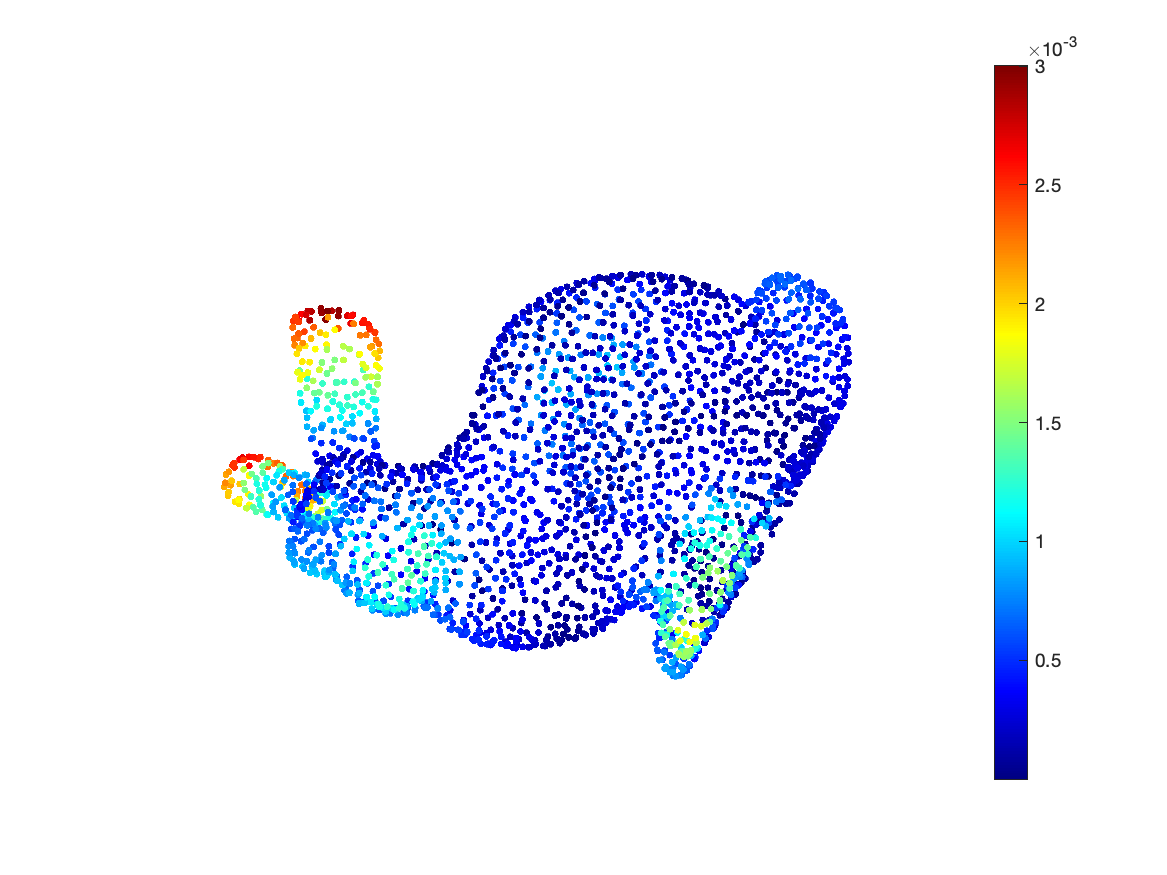}
    \\(g) $\calQ\calE_{\bL}(\f,\bq)=4.0\cdot10^{-6}$ for $B{=}2$
\end{minipage}\hspace{5pt}%
\begin{minipage}{0.23\textwidth}
    \centering
    \includegraphics[width=1.15\linewidth]{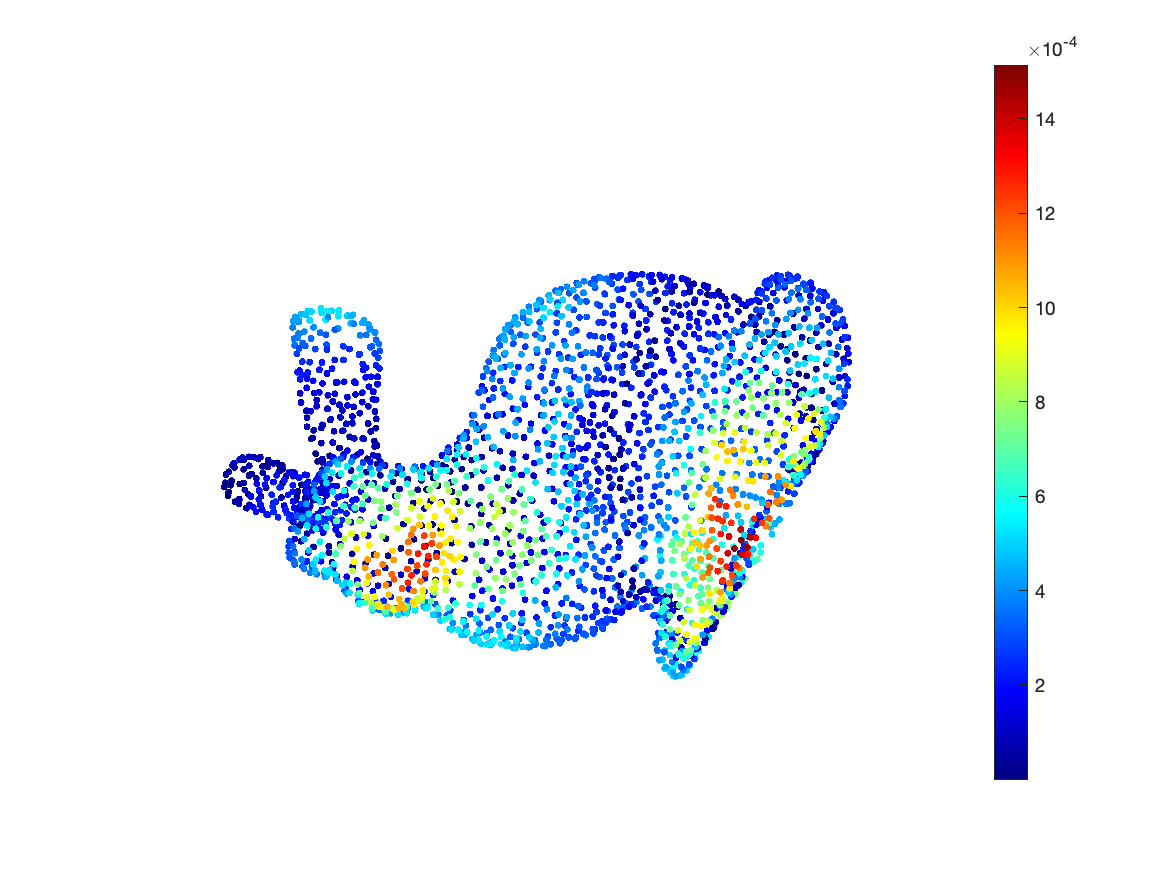}
    \\(h) $\calQ\calE_{\bL}(\f,\bq)=1.2\cdot10^{-6}$ for $B{=}4$
\end{minipage}

\caption{Comparison of the Bunny data and this quantized version: (a) original, (b–d) quantized versions, (e) data frequency spectrum, and (f–h) quantization errors. Each column corresponds to a specific method.}
\label{fig:Bunny}
\end{figure*}








We present numerical experiments to evaluate the performance of the proposed Single-Shot Noise Shaping (SSNS)
method and to compare it with existing quantization schemes. The experiments examine the effect of graph
topology, bandwidth, and bit budget on reconstruction accuracy, verify our theory,
and demonstrate its practical behavior on a 3D shape halftoning task.

While Algorithm \ref{alg:Preprocessing} is well suited to convey the concept of our noise shaping approach, it requires $\calO(r^3N)$ operations in the worst case. \citet{meyer2024thesis} developed a more efficient implementation of Algorithm \ref{alg:Preprocessing} reducing this to $\calO(r^2N)$ operations. In our experiments, we will exclusively use the superior version of \citet{meyer2024thesis}, which is summarized in Algorithm \ref{alg:PreprocessingII} in Appendix \ref{app:Algorithm}, and refer to it as \emph{Single Shot Noise Shaping (SSNS)}.

Our experiments were conducted on a MacBook Pro equipped with an Apple M1 processor and 16 GB RAM, using MATLAB R2023a. Graph signal processing operations, including graph Fourier transform computations and graph visualizations, were implemented using the GSPBOX toolbox \cite{perraudin2014gspbox}. The corresponding code is available in the supplementary materials. 


\begin{figure}
    \centering
    \includegraphics[width=\linewidth]{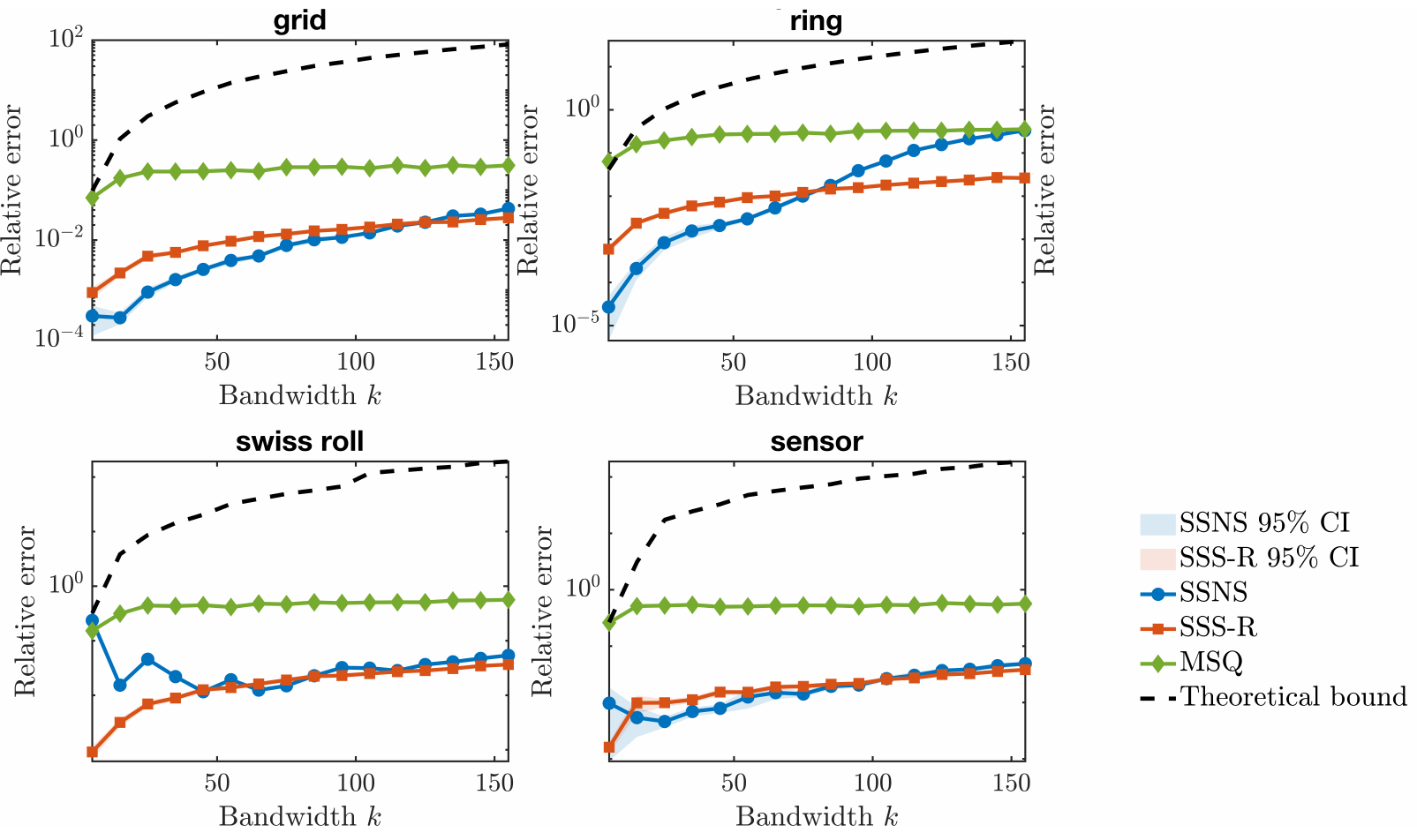}

    \caption{Comparison of the performance of SSNS to the theoretical worst-case guarantees \ref{thm:ErrorBound}, \emph{Memory-less quantization (MSQ)}, and to the \emph{Step-by-Step-Serving with Replacement (SSS-R)} method.}
    \label{fig:theory-SSNS-SSS-R}
\end{figure}

\textbf{Experiment 1: Quantization Error vs. Bandwidth Across Graph Topologies.} We begin by evaluating the performance of SSNS on several standard graph topologies, including \emph{grid}, \emph{bunny}, \emph{ring}, \emph{Swiss-roll}, \emph{sensor} and \emph{Minnesota} graphs. For each graph, we consider low-pass graph data $\f = \bX_r\boldsymbol{\alpha}$ of varying bandwidths $r \in [15,155]$, generated with Gaussian $\boldsymbol{\alpha} \in \R^r$, which is normalized to have $\| \f \|_\infty = 1$. 

We perform experiments with multiple bit budgets, $B \in \{1, 2, 4\}$, using $20$ independent realizations of the data for each bandwidth. SSNS is applied to each realization, and the quantization error is measured by the \emph{Relative Error} in $\ell_2$-norm after applying a brick-wall filter or bandwidth $r$, i.e.,
\begin{equation}\label{eq:num:rel:error}
\text{Relative Error} = \frac{\calQ\calE_{\bL_r}(\f,\bq)}{\| \f \|_2} = \frac{\| \bL_r(\f-\bq) \|_2}{\| \f \|_2}
\end{equation}

Figure~\ref{fig:preML_semilog} shows the Relative Error of SSNS as a function of the bandwidth for each graph, on a semi-logarithmic scale and averaged over the random realizations. Each curve corresponds to a different bit budget ($B = 1,2,4$). We observe the following:

\begin{itemize}[noitemsep, topsep=0.5pt]
    \item As the number of bits increases, the average relative error decreases consistently across all graphs, confirming the expected trade-off between bit budget and quantization error. We further investigate this in Experiment~3 in Appendix \ref{app:Experiments}.
    \item Overall, lower bandwidth induces lower quantization errors, due to low-bandwidth data being concentrated in fewer Fourier modes. However, we note that on certain graphs such as bunny or sensor this trend is reversed for very small bandwidths. So far, we could not find a satisfactory explanation of this peculiar behavior; it is also not reflected in our worst-case error bound, cf.\ Experiment 3. However, we suspect that this is related to the structure of the first eigenvectors, which are more dominant for low-bandwidth signals.
\end{itemize}

Figure \ref{fig:Bunny} compares the quantized representations $\f_{\bq} = \bL_r \bq$ of an exemplary data sample on the bunny graph with their unquantized counterpart. The figure visualizes how well low-bandwith data on graphs can be preserved by coarsely quantized samples. Furthermore, Figure \ref{fig:Bunny} (f)--(h) suggests that, under our single-shot noise shaping approach, the quantization error accumulates at ``high-curvature'' regions of the graph. It would be desirable to understand this effect in future work.


\textbf{Experiment 2: SSNS vs. SSS-R and Worst-Case Error Guarantees.} In the second experiment, we compare the performance of SSNS to our theoretical worst-case guarantees in Theorem~\ref{thm:ErrorBound}, to naive MSQ directly applied to the data, and to the \emph{Step-by-Step-Serving with Replacement (SSS-R)} method of \citet{krahmer2023quantization,krahmer2026quantization} 
under a fixed quantization budget of \(\log\log N\) bits required by the SSS-R method, cf.\ Section \ref{sec:Comparison}.\footnote{We did not include LGFF-based quantization \citep{reingruber2025efficient} in Experiment 2. Their method is hard to compare directly to ours since it produces overcomplete representations $\bq \in \calA^M$, for $M \ge N$, and operates in the transform domain with bit allocation, where informative coefficients receive more bits, while less relevant coefficients receive fewer bits.} 
Figure~\ref{fig:theory-SSNS-SSS-R} reports the Relative Error (on logarithmic scale) as a function of \(r\) for several representative graph topologies, including grid, ring, and Swiss-roll graphs. For each bandwidth value, the reported error is obtained by averaging over multiple random realizations (95\% confidence intervals are provided).

We compare our proposed \emph{SSNS} method with the baseline SSS-R scheme. Across all graph models, SSNS consistently attains lower or comparable average relative error for small to moderate bandwidths. As the bandwidth increases, the error of both methods grows, reflecting the increased incoherence of the graph Fourier basis associated with larger neighborhoods. For larger bandwidths, SSS-R shows more stability and achieves lower error. Let us mention that SSS-R only requires $\calO(rN \log(N))$ operations, which makes it more efficient for higher bandwidths. However, SSS-R is not available for 1-bit quantization and always requires $N \log(N)$-bits on average. 

The dashed curves correspond to the theoretical upper bound in \eqref{eq:theor-bound-SSNS}. In particular, we see that our theoretical bound is overly pessimistic due to its worst-case nature, but captures the overall growth trend with respect to \(r\) and provides theoretical justification for the observed performance degradation at larger bandwidths. Importantly, SSNS remains well below the theoretical bound across all considered graphs, confirming its robustness to quantization effects.



\textbf{Experiments --- Synopsis.} All conducted experiments confirm the effectiveness of our method for quantizing graph data with low bandwidth. The reduction in relative error with increasing bit budget is consistent across all tested graphs, demonstrating the method's robustness to both signal structure and graph topology. Furthermore, our results provide insights into the performance scaling of SSNS. Namely, graphs with more uniform connectivity (e.g., 2D grid) tend to exhibit smoother error growth compared to irregular graphs (e.g., sensor or bunny), reflecting the influence of graph~topology~on~the~distribution~of~Fourier~energy.  Additional experiments are provided in Appendix \ref{app:Experiments}, including an application of our method to 3D digital halftoning, and comparisons of the runtime and the empirical and theoretical bid-depth scaling of the error.


\section{Discussion}
\label{sec:Discussion}

We have shown in this paper that Single Shot Noise Shaping (SSNS) via Algorithm \ref{alg:Preprocessing} provides a simple and effective approach to quantizing graph signals with low computational complexity and predictable scaling behavior. While it is computationally more expensive than existing noise shaping methods such as Step-by-Step-Serving with Replacement (SSS-R), it comes with improved theoretical guarantees, allows for a flexible choice of bit-levels including extreme one-bit quantization, and shows better performance for very small bandwidths.

\textbf{Limitations.} In our study, we focused on the brick-wall filter as a prototype of low-pass filters. Various other low-pass filters exist, such as  \cite{defferrard2016convolutional, chang2021not, nt2019revisiting}. Given that the brick-wall filter is the standard example of an idealized low-pass filter, and that to the best of our knowledge this is the first approach to graph quantization that allows one-bit quantization \emph{and} comes with rigorous error guarantees in the $\mathcal Q \mathcal E_L$-measure, we believe that such a restriction is more than legitimate. We furthermore provided no robustness analysis of our method on approximately $r$-bandlimited data, and only use the graph incoherence to characterize the influence of the graph topology. As our experiment in Figure \ref{fig:Bunny} (f)--(h) suggests, a more refined analysis might be needed to fully explain the observed quantization error on complex graphs for strongly bandlimited signals, cf.\ Experiment 1 in Section \ref{sec:NumericalExperiments}.

\bibliography{paper}
\bibliographystyle{plainnat}


\newpage

\appendix

\section{Supplementary material}

\subsection{Efficient implementation of Algorithm \ref{alg:Preprocessing}}
\label{app:Algorithm}

We provide pseudo-code of the accelerated preprocessing algorithm of \citet{meyer2024thesis} in Algorithm \ref{alg:PreprocessingII}. It improves over Algorithm \ref{alg:Preprocessing} by extracting kernel vectors $\bb$ in a more efficient way. To this end, it splits the columns of $\bX$ into blocks of size $r$ and iteratively extracts collections of $r$ kernel vectors per block. By slightly modifying the single kernel vectors in $r$ inner iterations, they can all be used before a new collection has to be generated. Generating $r$ kernel vectors is of the same complexity as generating a single kernel vector. Asymptotically, each of the $\lceil N/r \rceil$ outer iterations of Algorithm \ref{alg:PreprocessingII} thus requires as many operations as each of the $N$ iterations of Algorithm \ref{alg:Preprocessing}, leading to an overall complexity improvement of a factor $r$.

Note that in our presentation of Algorithm \ref{alg:PreprocessingII} we assume that $N$ is divisible by $r$ for convenience. If this is not the case, the last iteration of the outer loop has to be slightly modified since less than $r$ linearly independent vectors $\bb_i^{(1)}$ remain in $\ker_{I^c}(\bX)$ in the last step.

\begin{algorithm}[t]
   \caption{Vector preprocessing \cite{meyer2024thesis}}
   \label{alg:PreprocessingII}
    \begin{algorithmic}
       \STATE {\bfseries Input:} $\bX \in \R^{r\times N}$ ($r < N$), $\bz_0 \in \R^N$, and $c \ge \|\bz_0\|_\infty$
       \STATE 
       \STATE Initialize $k=0$ and  $$J_0 = \{ i \in [N] \colon \text{the $i$-th column of $\bX$ is zero} \}.$$
       \STATE Define $\bb \in \R^N$ via $\bb_{J_0^c} = \boldsymbol{0}$ and $b_i = c - (z_0)_i$, for $i\in J_0$, such that $\bb \in \ker(\bX)$ and $|(z_0)_i + b_i| = c$ for $i \in J_0$.
       Replace $\bz_0$ with $\bz_0 + \bb$.

       \REPEAT
       \STATE Let $I = \{i_1,\dots,i_{2r}\} \subset J_k^c$ contain the $2r$ smallest indices in $J_k^c$.
       \STATE Compute a set of $r$ linearly independent vectors $\calB^{(1)} = \{\bb_1^{(1)},\dots,\bb_{r}^{(1)}\}$ in $\ker_{I^c}(\bX)$.
       \STATE $\bz_{k,1} \leftarrow \bz_k$
       \STATE $J_{k,1} \leftarrow J_k$
       \FOR{$\ell=1,\dots,r$}
       \STATE $\bb \leftarrow \bb_1^{(\ell)}$
       \STATE Compute $\alpha \in \R$ with $\| \bz_{k,\ell} + \alpha \bb \|_\infty = c$
       \STATE $\bz_{k,\ell+1} \leftarrow \bz_{k,\ell} + \alpha \bb \in \R^{N}$
       \STATE $i_* \leftarrow \argmin \{ i \in I \colon |(\bz_{k,\ell+1})_i| = c \}$
       \STATE $J_{k,\ell+1} \leftarrow J_{k,\ell} \cup \{ i_\star \}$
       \STATE Define $\calB^{(\ell+1)} = \{ \bb_1^{(\ell+1)},\dots,\bb_{r-\ell}^{(\ell+1)}\}$ with
       \begin{align*}
            \bb_j^{(\ell+1)} = \bb - \frac{b_{i_\star}}{(b_{j+1}^{(\ell)})_{i_\star}} \bb_{j+1}^{(\ell)},
       \end{align*}
       which consists of $(r-\ell)$ linearly independent vectors with support in $I \setminus J_{k,\ell+1}$.
       \ENDFOR
       \STATE $J_{k+1} \leftarrow J_{k,r+1}$
       \STATE $\bz_{k+1} \leftarrow \bz_{k,r+1}$ 
       
       \STATE $k \leftarrow k+1$
       \UNTIL{$\| |\bz_k| - c\boldsymbol{1} \|_0 \le r $}
       \STATE $k_{\text{final}} = k$
	    \STATE 
		\STATE \textbf{Output:} $\bz_{k_{\text{final}}}$ for which $\bX \bz_{k_{\text{final}}} = \bX\bz_0$, $\| \bz_{k_{\text{final}}} \|_\infty = c$, and $\| |\bz_{k_{\text{final}}}| - c \boldsymbol{1} \|_0 \le r$
    \end{algorithmic}
\end{algorithm}

\subsection{Additional numerical experiments}
\label{app:Experiments}

\textbf{Experiment 3: Empirical and Theoretical Bit-Depth Scaling.} 
In this experiment, we investigate the effect of quantization bit-depth on the reconstruction accuracy of bandlimited graph signals under the SSNS method, and compare its empirical performance with our theoretical error bounds. We consider the same graph topologies as in Experiment 1. 
For each graph, we generate bandlimited graph signals supported on the first $r=200$ graph Fourier modes, and generate $50$ independent signal realizations.
We measure the relative squared reconstruction error \eqref{eq:num:rel:error} and report the average error over all realizations for each bit-depth.

The theoretical analysis predicts that the reconstruction error scales as
$
\mathcal{O}\!\left( 2^{-B}  \right)
$ if $r$ and $N$ are fixed. In Figure \ref{fig:BitDepth}, we display this bound (up to a constant factor) together with the average relative error. We see that for small bit rates the actual average error is considerably smaller than expected. Note that this does not contradict our error bounds which are based on worst-case estimates.

\begin{figure}[t]
    \centering   \includegraphics[width=0.6\linewidth]{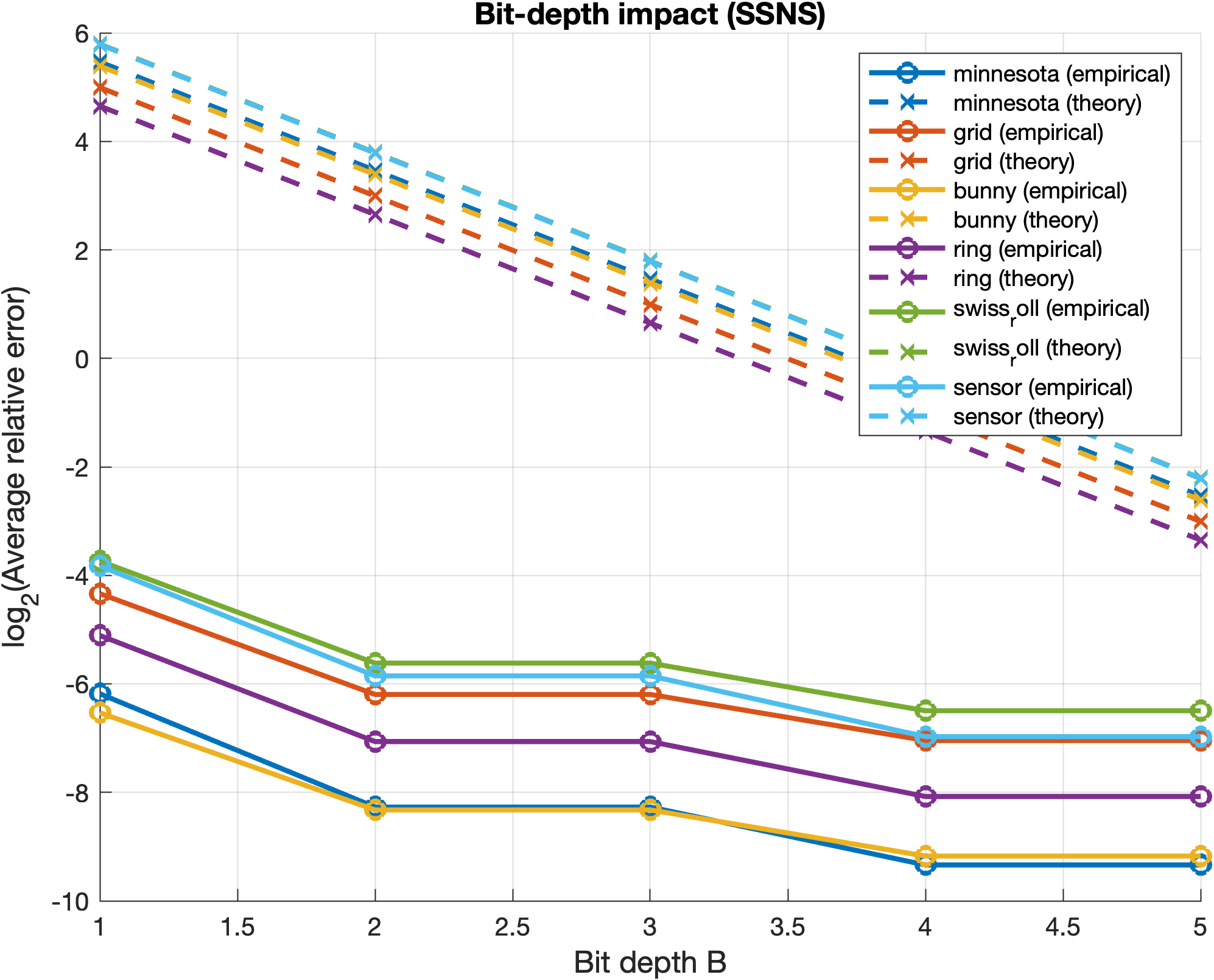}
    \caption{Average relative error versus bit-depth for SSNS across different graph topologies, compared with the theoretical bound proportional to $2^{-B}$. Results are averaged over 50 random signals with bandwidth $r=200$.}

    \label{fig:BitDepth}
\end{figure}

\textbf{Experiment 4: Wall-clock Runtime Comparison.} We evaluate the wall-clock runtime of the two quantization strategies on representative graph datasets, namely \texttt{grid}, \texttt{bunny}, and \texttt{swiss\_roll}. 
For each graph, we generate bandlimited signals with a fixed bandwidth of $r = 50$ and perform quantization using the two approaches \textit{SSNS} and  \textit{SSS-R}. The experiments are presented in Figure~\ref{Fig7:wall-clock}. All experiments are repeated over $5$ independent runs, and we report the average runtime.

We explicitly separate the computational cost of the graph Fourier basis (eigendecomposition) from the quantization step. 
Figure \ref{Fig7:wall-clock} shows that the quantization costs of SSS-R (excluding the eigendecomposition) are notably lower than the ones of SSNS. However, for large graphs, the eigenbasis computation dominates the total runtime of both algorithms.

\begin{figure*}[t]
    \centering
    \includegraphics[width=0.32\linewidth]{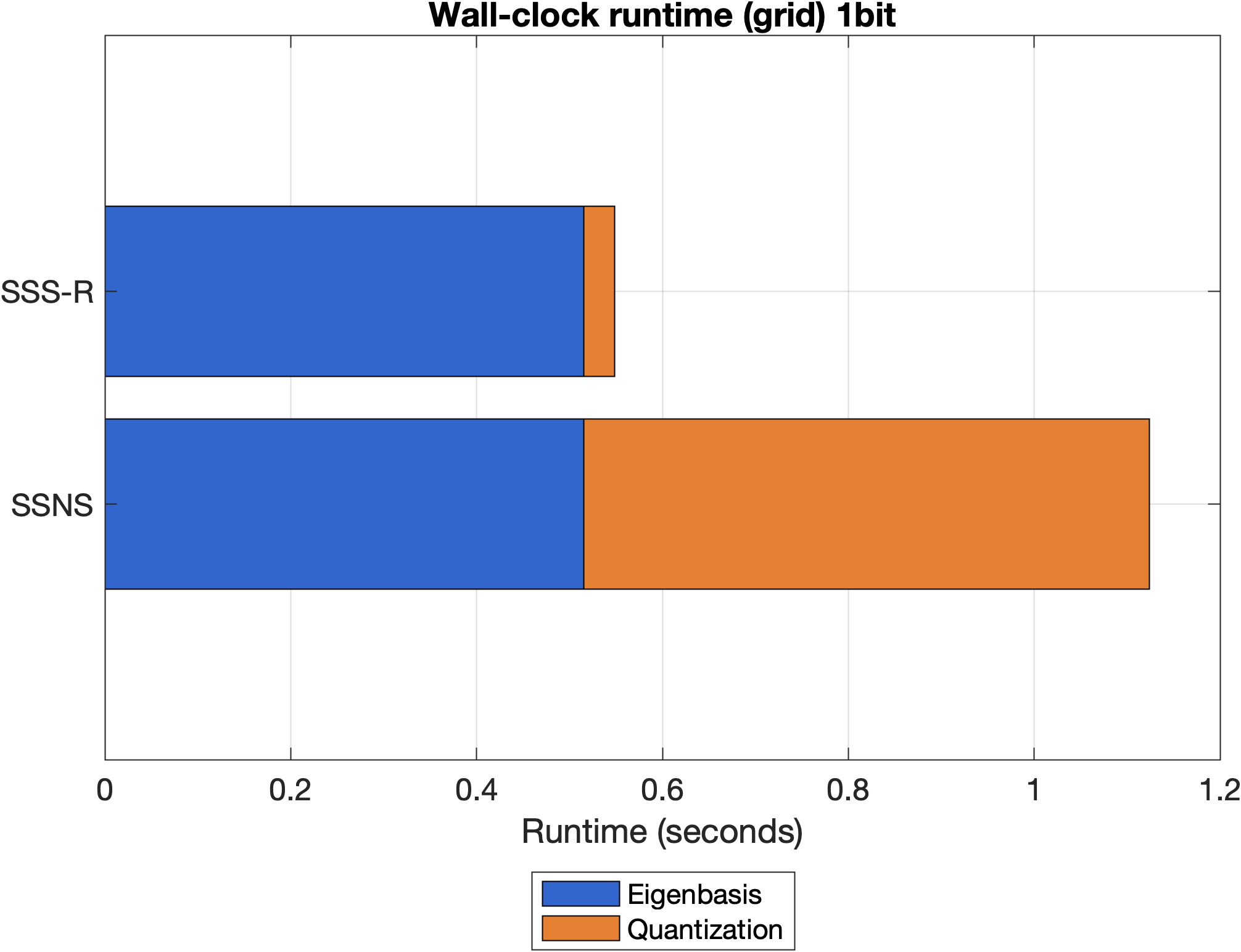}
    \includegraphics[width=0.32\linewidth]{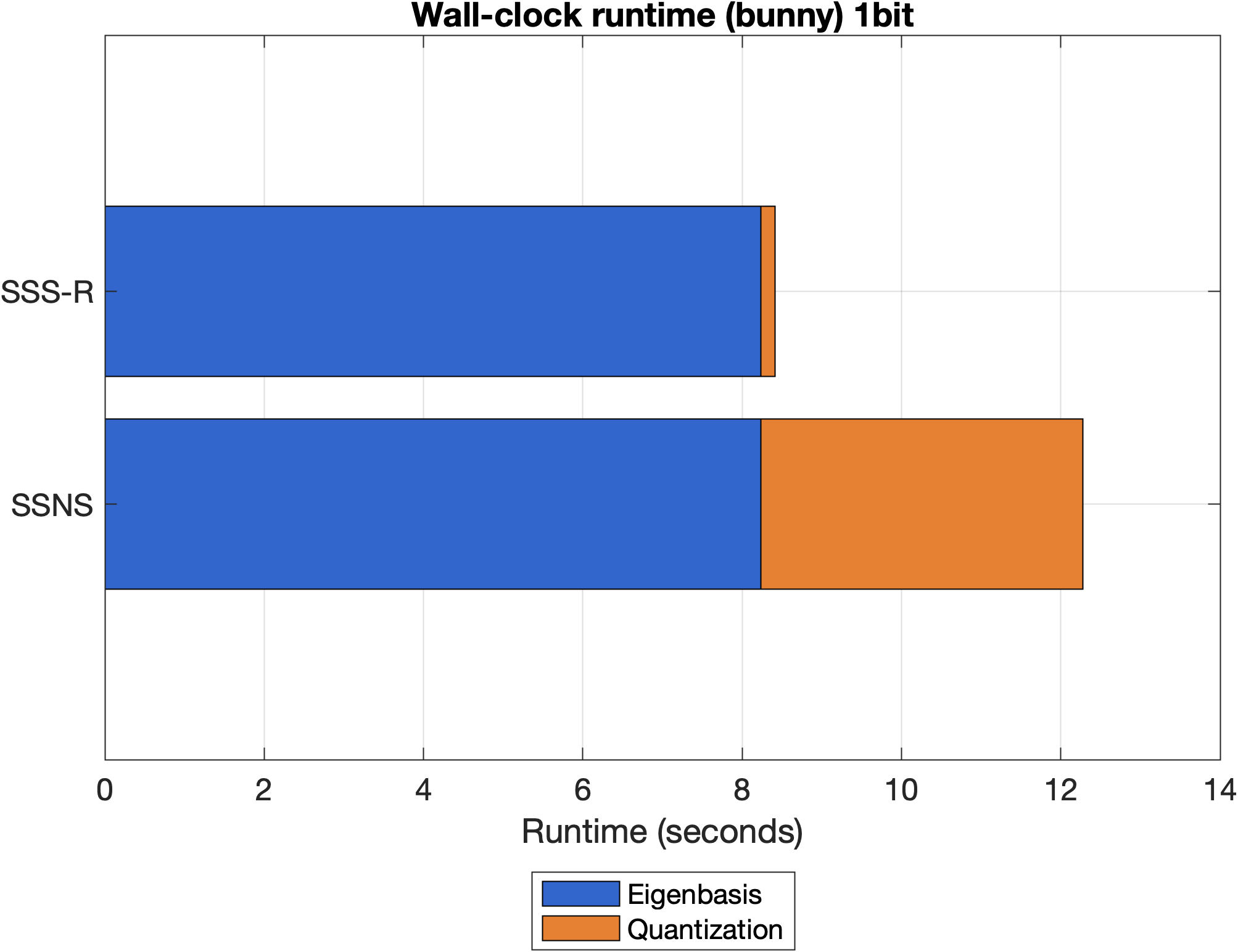}
    \includegraphics[width=0.32\linewidth]{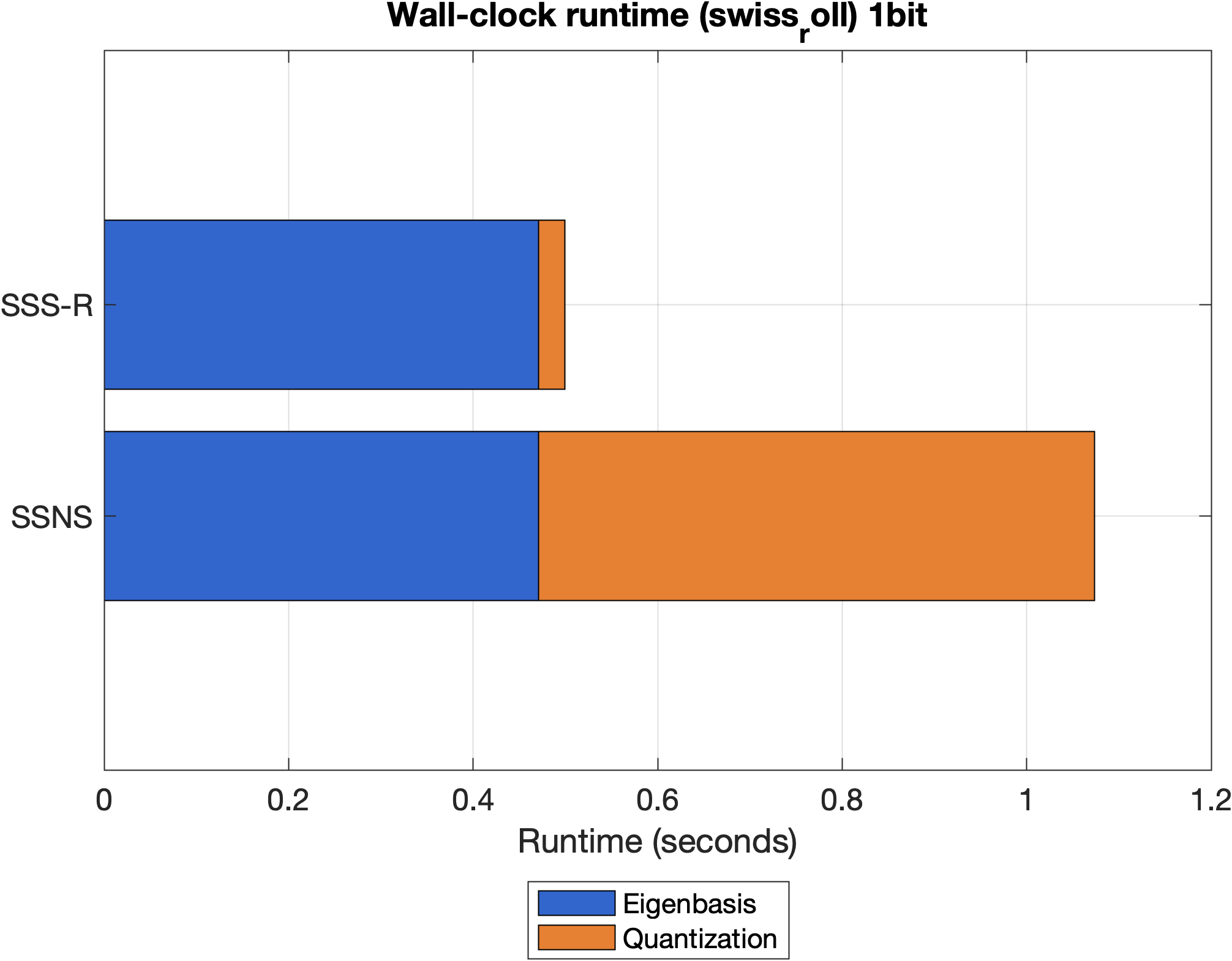}
    \caption{Wall-clock runtime comparison for different graph types (1-bit quantization). 
    Each bar shows the total runtime, decomposed into eigenbasis computation (blue) and quantization (orange).}
    \label{Fig7:wall-clock}
\end{figure*}

\textbf{Experiment 5: Graph-Signal Halftoning on 3D Meshes.}
Finally, we evaluate the proposed \emph{SSNS} algorithm for halftoning of graph signals on 3D shapes and compare it with Sigma-Delta-Weight (SDW) in \cite{krahmer2023quantization} and \emph{memoryless scalar quantization (MSQ)}; see also related work on
3D halftoning in \cite{gooran20203d, gooran2022three, mao20173d, lou1998fundamentals, abedini2023multi, abedini20232d}.

We load the Stanford Bunny mesh using PyVista and construct an undirected graph over its vertices via a
$k$-nearest-neighbors rule, similar to the procedure used in GSPBox. The resulting graph Laplacian defines the
graph structure. As a test signal, we consider a non-bandlimited graph signal $\f {\in} [0,1]^N$ given by the
$z$-coordinate of each vertex, where $N$ is the number of vertices. We then compute 1~-~bit halftoned signals
$\bq{\in}\{-1,1\}^N$ using MSQ, SDW, and the proposed SSNS method, and directly visualize $q$ on the mesh. For  both SDW and SSNS, we conduct experiments with two heuristic bandwidth rates, namely~$r=20$~and~$r=50$,~see~Fig.~\ref{fig:BunnyHalftoning}

This experimental design follows the classical motivation of digital halftoning: although the displayed signal is binary, perceptual integration by the human visual system acts as an implicit low-pass filter, so visual appearance is primarily determined by low-frequency content
\cite{SteinbergFloyd,lyu2023sigma,
krahmer2022enhanced}.
While the signal is not bandlimited and therefore lies outside the scope of our theoretical analysis, the three
methods produce visually distinct halftoning patterns, illustrating different ways of shaping quantization
artifacts on graphs. We emphasize that with higher bandwidth $r$, among the compared approaches, only SSNS is supported by provable error bounds, highlighting its role as a theoretically justified alternative for graph-based halftoning.

\begin{figure*}[t] 
\centering
\begin{minipage}{0.32\textwidth}
    \centering
    \includegraphics[width=1.1\linewidth]{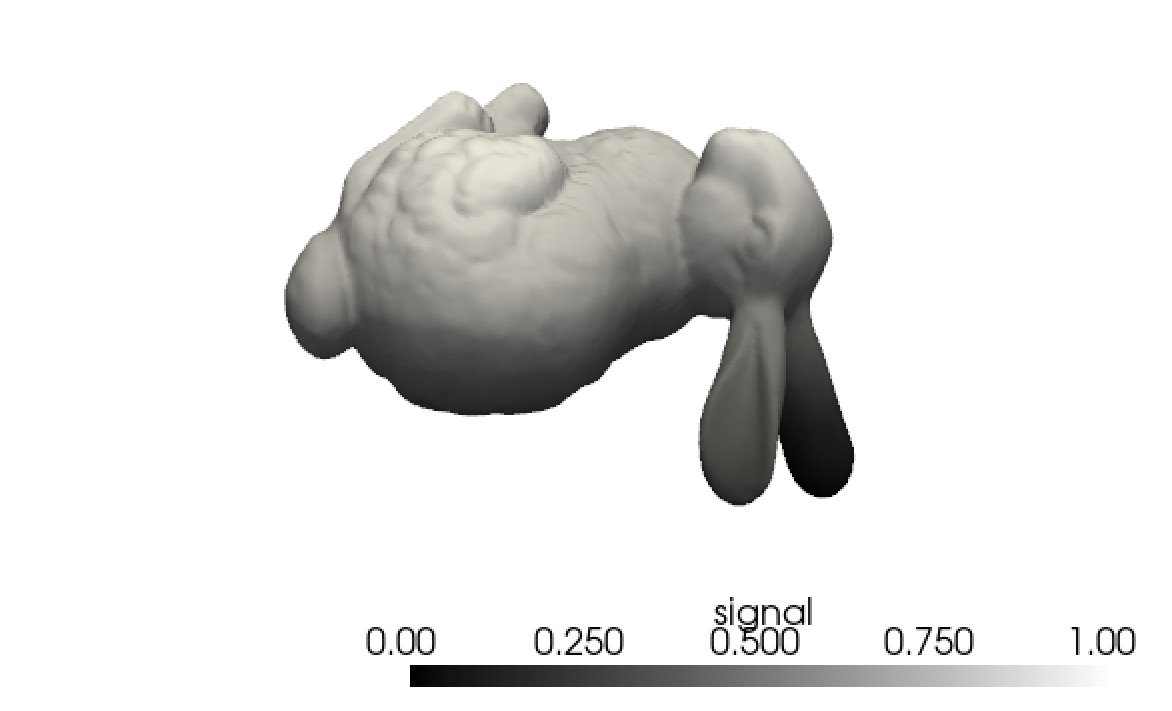}
    \\(a) Original bunny shape
\end{minipage}\hspace{5pt}%
\hspace{5pt} 
\begin{minipage}{0.31\textwidth}
    \centering    \includegraphics[width=1.1\linewidth]{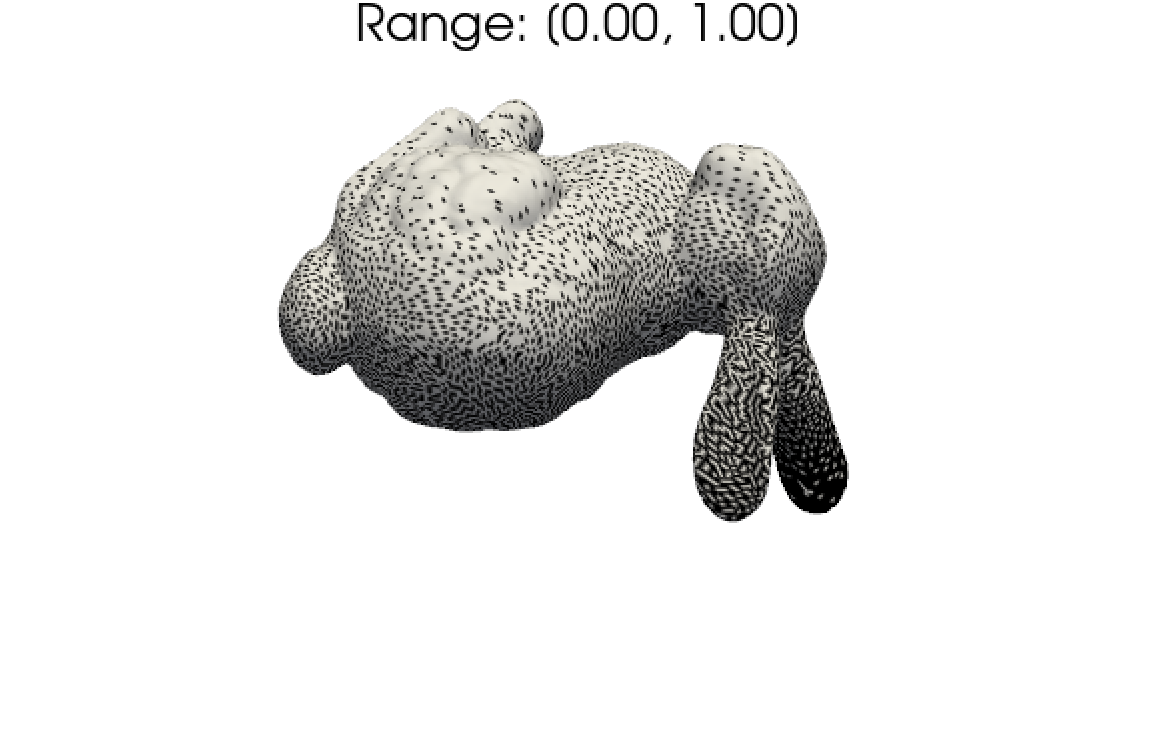}
    \\ (c) Bunny shape, 1 bit SDW, $r=20$ 
\end{minipage}\hspace{5pt}%
\begin{minipage}{0.32\textwidth}
    \centering  \includegraphics[width=1.1\linewidth]{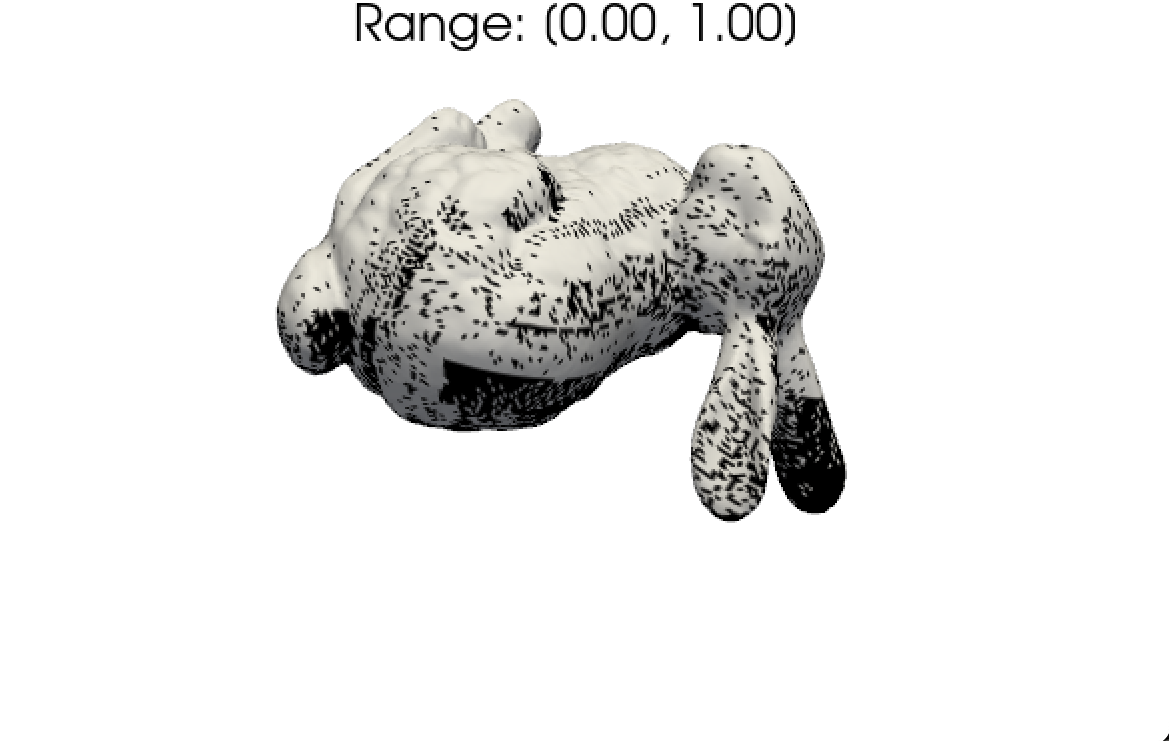}
    \\(d) Bunny shape, 1 bit SSNS,  $r=20$
\end{minipage}
\hspace{5pt}%
\hfill

\vspace{2mm}
\begin{minipage}{0.32\textwidth}
    \centering
    \includegraphics[width=1.1\linewidth]{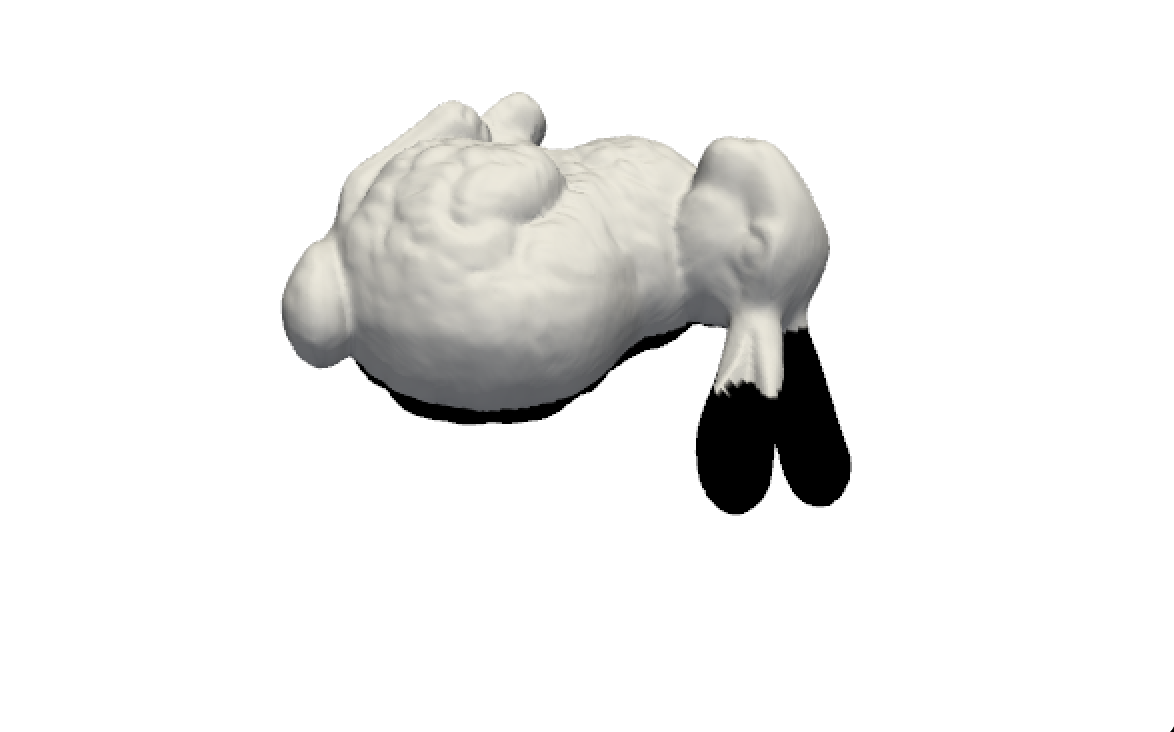}
    \\(b) Bunny shape, 1 bit MSQ  
\end{minipage}
\hspace{5pt}%
\begin{minipage}{0.32\textwidth}
    \centering    \includegraphics[width=1.1\linewidth]{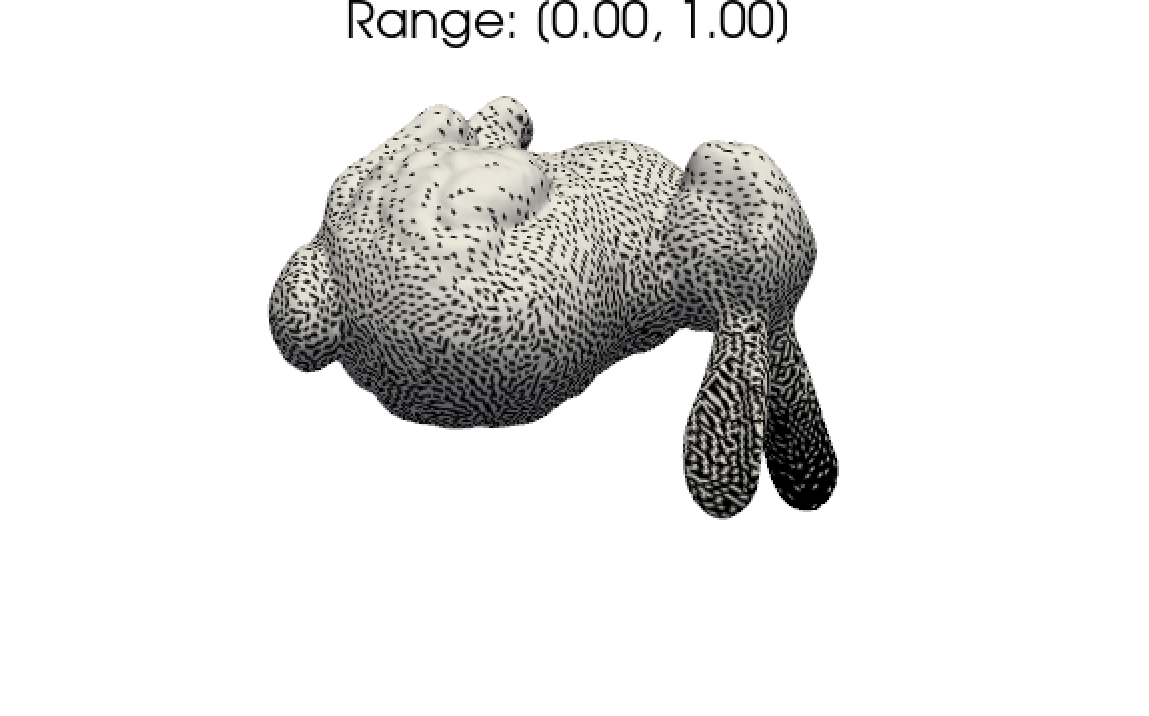}
    \\ (e) Bunny shape, 1 bit SDW, $r=50$
\end{minipage}\hspace{5pt}%
\begin{minipage}{0.32\textwidth}
    \centering  \includegraphics[width=1.1\linewidth]{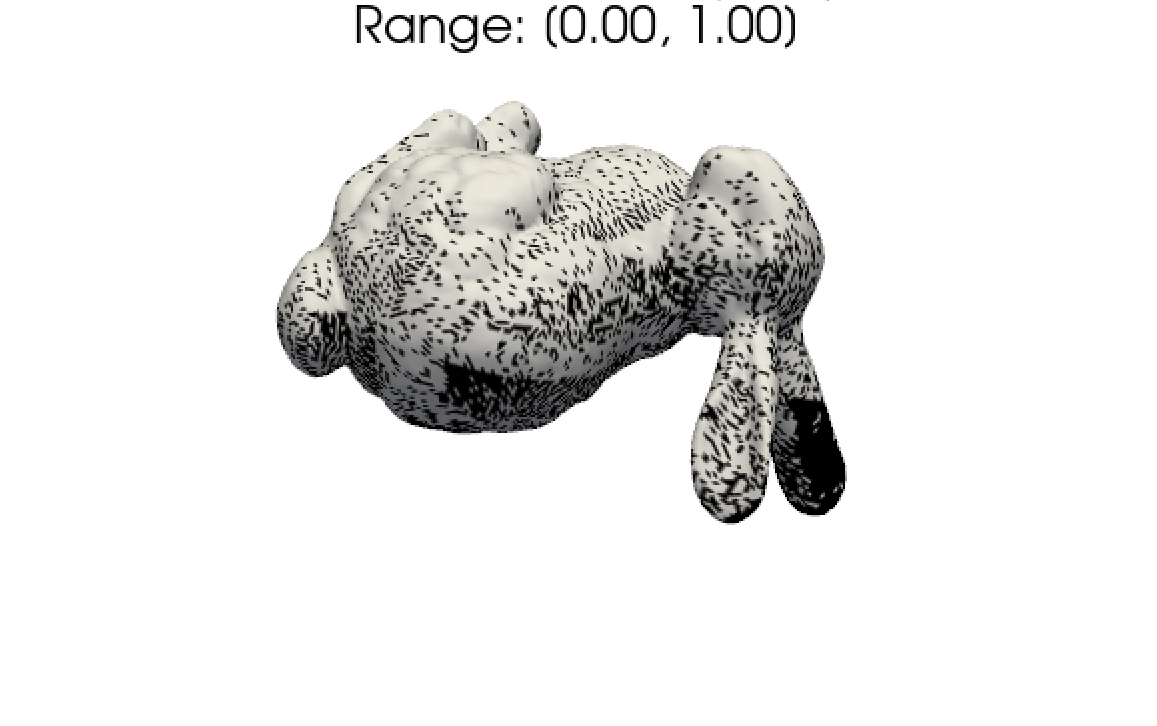}
    \\(f) Bunny shape, 1 bit SSNS $r=50$
\end{minipage}

\caption{Halftoning of the Stanford Bunny mesh graph signal using 1-bit quantization.
(a) Original signal.
(b) Memoryless scalar quantization (MSQ).
(c),(e) Sigma-Delta-Weight  (SDW) \cite{krahmer2023quantization} with different bandwidth $r$.
(d),(f) Proposed Single-Shot Noise Shaping (SSNS) with different different bandwidth  $r$.
All halftoned outputs are binary signals $\bq \in \{-1,1\}^N$ visualized on the mesh.}
\label{fig:BunnyHalftoning}
\end{figure*}


\newpage

\end{document}